\documentclass[11pt]{article}
\usepackage{amsmath,amsthm,amsfonts,amssymb,mathrsfs,bm}
\usepackage{xspace}
\usepackage[dvipsnames]{xcolor}
\usepackage{natbib}
\usepackage{authblk}
\usepackage[margin=1in]{geometry}
\usepackage{txfonts}
\usepackage[T1]{fontenc}
\usepackage[
pagebackref,
colorlinks=true,
urlcolor=blue,
linkcolor=blue,
citecolor=OliveGreen,
]{hyperref}
\usepackage[nameinlink]{cleveref}

\newcommand{\xhdr}[1]{\paragraph*{%\hspace*{-0.16in}
  \bf {#1}}}
\newcommand{\omt}[1]{}

\def\eps{{\varepsilon}}
\def\D{{\bf D}}
\def\sneq{{\tiny \mbox{$\neq$}}}
\def\subsetneq{\ \lower.5ex\hbox{$\stackrel{\subset}{\small \sneq}$}\ }

\newcommand{\reals}{\mathbb{R}}

\newcommand{\expect}{\mathbb{E}}
\newcommand{\given}{\,|\,}

\newcommand{\defeq}{\stackrel{\Delta}{=}}

\newcommand{\ialloc}[1]{f_{#1}}

\newtheorem{thm}{Theorem}
\newtheorem{prop}{Proposition}
\newtheorem{lem}{Lemma}
\newtheorem{cor}{Corollary}

\theoremstyle{definition}
\newtheorem{defn}{Definition}

\Crefname{thm}{Theorem}{Theorems}
\Crefname{lem}{Lemma}{Lemmas}

\begin{document}
% Title portion. Note the short title for running heads
\title{Delegated Search Approximates Efficient Search}
\author{Jon Kleinberg}
\author{Robert Kleinberg}
\affil{Cornell University, Ithaca, NY 14853, USA.\footnote{%
Email: \texttt{\{kleinber,rdk\}@cs.cornell.edu}.}}
\date{}
%   \affiliation{%
%   \institution{Cornell University}
%   % \streetaddress{}
%   \city{Ithaca}
%   \state{NY}
%   \postcode{14853}
% % \email{email}
%
% \author{Robert Kleinberg}
%   \affiliation{%
%   \institution{Cornell University}
%   % \streetaddress{}
%   \city{Ithaca}
%   \state{NY}
%   \postcode{14853}
% \email{email}

\maketitle

% note that the abstract must come before \maketitle
\begin{abstract}
% This is an EC 2018 submission about mitigating bias in delegated search.
There are many settings in which a principal performs a task by
{\em delegating} it to an agent, who searches over possible solutions
and proposes one to the principal.
This describes many aspects of the workflow within organizations,
as well as many of the activities undertaken by regulatory bodies,
who often obtain relevant information from the parties being regulated
through a process of delegation.
A fundamental tension underlying delegation is the fact that
the agent's interests will typically differ -- potentially significantly --
from the interests of the principal, and as a result the agent
may propose solutions based on their own incentives that are 
inefficient for the principal.
A basic problem, therefore, is to design mechanisms by which the
principal can constrain the set of proposals they are willing
to accept from the agent, to ensure a certain level of quality
for the principal from the proposed solution.

In this work, we investigate how much the principal loses --
quantitatively, in terms of the objective they are trying to optimize --
when they delegate to an agent.
We develop a methodology for bounding this loss of efficiency,
and show that in a very general model of delegation,
there is a family of mechanisms achieving 
a universal bound on the ratio between 
the quality of the solution obtained through delegation and
the quality the principal could obtain in an idealized benchmark
where they searched for a solution themself.
Moreover, it is possible to achieve such bounds through mechanisms
with a natural threshold structure, which are thus structurally simpler
than the optimal mechanisms typically considered in the literature
on delegation.
At the heart of our framework is an unexpected connection between
delegation and the analysis of {\em prophet inequalities}, which 
we leverage to provide bounds on the behavior of our delegation
mechanisms.

\end{abstract}

% note: this command has been disabled to remove the ACM copyright block. Sorry...
%\thanks{This work is supported by the National Science Foundation,
%  under grant CNS-0435060, grant CCR-0325197 and grant EN-CS-0329609.}

% <concept_id>10003752.10003809</concept_id>
% <concept_desc>Theory of computation~Design and analysis of algorithms</concept_desc>
% <concept_id>10010405.10010455.10010460</concept_id>
% <concept_desc>Applied computing~Economics</concept_desc>
% \ccsdesc[500]{Theory of computation~Design and analysis of algorithms}
% \ccsdesc[500]{Applied computing~Economics}
%
% \keywords{Delegated search, prophet inequalities.}

\section{Introduction}
\label{intro}

There are many settings in which a decision-maker is faced with a
difficult problem that they cannot solve on their own,
and so they instead approach it in two steps:
they first {\em delegate} the search for possible solutions
to an {\em agent} who is able to invest more time in the process,
and then they evaluate the solution(s) that the agent proposes.
One concrete example arises in organizations or firms, where
the management may delegate the search for solutions to a
division that reports to them, ultimately
making a decision on the solution that is proposed by the division
\citep{aghion-formal-real-authority,armstrong-delegated,li-delegation-organizations}.
A second example arises in regulation, where a governmental
agency needs to decide whether there is a way to structure
a proposed corporate merger in a way that is compatible with
regulatory guidelines; the companies seeking to merge
study possible ways of structuring the merger,
propose one or more to the regulator, and seek the regulator's approval.
In this way, the search over structures for the merger
has implicitly been delegated from the regulator to the companies,
with the regulator retaining decision-making control over the
the proposed solution \citep{farrell-welfare-antitrust,nocke-dynamic-merger}.
A similar scenario could be described for regulation in other
settings, where a company may be searching over possible
solutions that comply with environmental law, employment law,
or other guidelines.

The interesting tension in all these situations is that the
decision-maker who delegates the task (henceforth referred to
as the ``principal'') has a particular objective
function that they are seeking to optimize; but the agent who actually
performs the task might have interests that
are not directly aligned with the principal's.
For example, in the regulatory context, the regulator (acting
as principal) may reasonably suspect that a merger proposed
by a set of companies (acting as the agent) will be structured in
a way that strongly benefits the companies, even if other feasible
structures would have been better for the market or for society as a whole.
Similarly, a group within an organization tasked with solving
a problem may well preferentially search for a solution
that benefits them in particular.
Given this natural set of incentives,
how should the principal structure the delegation to
the agent so as to ensure that the solution the agent proposes
performs well under the principal's own objective function?

A rich literature has developed in economics around the
formalization and analysis of delegation, focusing
on this tension between the conflicting objectives of
the principal and the agent;
see \citet{holmstrom-thesis,holmstrom-theory-delegation}
for influential early research, and
\cite{alonso-optimal-delegation,armstrong-delegated,
amador-optimal-delegation,ambrus-delegation-money-burning}
for recent work.
A dominant theme in this line of work is that the principal does not
offer monetary compensation to the agent as a way of favoring
certain proposed solutions over others (though see
\citet{krishna-delegation-monetary}); this is consistent with
the motivating applications, in which for example regulators
in many contexts can accept or reject proposals from companies, but cannot
selectively offer varying amounts of compensation to these companies based
on the content of the proposal.
This lack of monetary transfers between the parties imparts
a fundamental structure to the problem, in which the principal
can simply define a mechanism implicitly specifying the subset of all
``eligible'' proposals that they are willing to accept,
and the agent is then incentivized to
search for solutions that are good for them but that also
lie in this eligible set.
A long line of work has gone into determining the structure of
eligible sets that produce optimal mechanisms for the principal,
yielding constructions that are often quite intricate
\citep{alonso-optimal-delegation,armstrong-delegated,melumad-delegation-no-transfer}.

\xhdr{The Present Work}
Given how broadly delegation is used across a range of contexts,
it is interesting to consider how precarious a process it is ---
the principal is ceding control of
their search problem to an agent whose interests might be
completely misaligned with their own, and the only leverage the
principal has is to accept or reject the solution that is eventually proposed.
How much does the principal give up ---
quantitatively,
measured in terms of the objective they are trying to optimize ---
when they delegate to an agent?
Is there some robust, intrinsic reason why things don't turn out
as badly as we might fear?
And how do the answers to these questions depend on what the
principal is actually able to observe about the agent's solution ---
including how much effort the agent put into finding the solution,
and how good it is not just for the principal but also for the agent?

In their most natural formulation,
these are inherently comparative questions, since they seek to
relate the solution quality obtained through delegation to
the solution quality in an alternate, ideal setting where
delegation was not necessary.
As such, they address an issue fundamentally distinct
from the primary focus of the existing literature on delegation,
which as noted above has centered on characterizing mechanisms that
produce optimal
delegation for the principal, without this type of comparative evaluation.

There is a natural benchmark to use for our comparison:
we could measure the quality of the outcome under delegation versus the
quality of the solution that the principal could obtain were
they to perform the search task themself,
investing the same level of effort in the search that the agent does.
Now, there are many settings where it may be too costly or otherwise
infeasible for the principal to actually conduct this search,
but this benchmark nonetheless
provides a conceptual reference point to make clear how much
payoff to the principal is lost through delegation.
In this sense, it plays the role of an optimal point of comparison,
much like the role of the intractable optimum in an approximation algorithm
or the societal optimum in a price-of-anarchy analysis.

In this paper we develop a methodology to bound the performance
of delegated search, relative to the benchmark
in which the principal searches for a solution on their own.
Our methodology builds on a set of links that we identify
between bounds on delegated search and
the analysis of some fundamental models of decision-making
under uncertainty --- in particular, a surprisingly strong connection
between delegated search and bounds on {\em prophet inequalities}.
The connections between these formalisms
turn out be quite natural and useful, but to our knowledge they
have not been previously identified in either of
the literatures on delegated search or on prophet inequalities.
This connection not only provides bounds on the quality of delegated
solutions relative to an ideal benchmark; it also shows that
strong bounds can be obtained using eligible sets that are
structurally very simple --- in a number of cases defined simply by
a carefully chosen threshold rule ---
and hence in contrast with the complex constructions associated
with optimal mechanisms.

\subsection{Overview of Models}

\xhdr{A Distributional Model}
We begin by describing the models in which we perform our analysis.
Our main model, which is essentially the one considered in
\citet{armstrong-delegated}, has the principal and the agent
agree that the agent will consider $n$ candidate solutions
and propose one to the principal;
$n$ thus represents the level of effort that the agent commits to the problem.
The principal will only see the solution that is proposed, not
the other $n-1$ that the agent also considers.

What does it means for the agent to consider a candidate solution?
We assume that the solutions belong to an abstract space
$\Omega$ with a probability measure on it, and the agent's
search for a solution consists of performing $n$ independent
and identically distributed draws from $\Omega$,
resulting in a set of candidate solutions
$\omega_1, \omega_2, \ldots, \omega_n \in \Omega$.\footnote{Later we will
also consider the case in which different draws by the agent
can come from different probability measures on $\Omega$;
for example, this can model the case in which the agent is a group
of $n$ employees in an organization, and the $i^{\rm th}$ solution
is drawn by the $i^{\rm th}$ agent, who may have a different distribution
over solutions from the $j^{\rm th}$ agent.}
Each solution $\omega_i$ drawn by the agent has a quality
for the principal, denoted $x(\omega_i)$, and a possibly different
quality for the agent, denoted $y(\omega_i)$.
The agent selects one of its candidate solutions, say $\omega_i$,
to present to the principal.
(Below, we will discuss the contrast between the model in which
the principal is able to determine both $x(\omega_i)$ and $y(\omega_i)$ ---
the value to both the agent and themself --- for the proposed
solution, and the model in which the principal is only able
to determine $x(\omega_i)$.)

Now, if the principal imposed no constraint on the agent's behavior,
then the agent would simply choose the solution $\omega_i$ that
maximizes $y(\omega_i)$, and the principal would receive whatever corresponding
$x(\omega_i)$ value resulted from this choice.
To improve on this, the principal could specify at the outset that they
will only accept $\omega$ values that satisfy some predicate on
$x(\omega_i)$ and (in the case that they can determine it) $y(\omega_i)$;
we will refer to the set of all $\omega$ satisfying the principal's
predicate as the {\em eligible set} of solutions.
It is thus in the agent's interest to propose a solution belonging
to the eligible set;
we ask whether one can design eligible sets that provide
provable bounds on the expected quality of the solution to the principal,
relative to the scenario in which the principal simply were to
draw $n$ times from $\Omega$ and select the sampled solution $\omega$
with maximum $x(\omega)$.

\xhdr{A Binary Model}
In our first, distributional model, the agent draws a set of
candidate solutions
$\omega_1, \ldots, \omega_n$ that the principal
cannot observe, and then must choose one to present to the principal.
This models a setting in which the agent explores a design space
and cannot fully anticipate which $\omega_i$ it will
encounter until it begins this exploration.
But we can also imagine settings where the principal and agent
both know that the set of potential options comes from a large discrete
set $\omega_1, \omega_2, \ldots, \omega_m$, and the only
question is which of these options is actually feasible to implement.
For example, there may be $m$ standard ways of structuring a merger
in a given industry, and the only question is which are possible
for the companies in question.

We model this version with publicly known binary options as follows.
There is a set of options $\omega_1, \omega_2, \ldots, \omega_m$,
and for each $i$, there is a known probability $p_i > 0$ such that
option $\omega_i$ is {\em feasible} with probability $p_i$, and
{\em infeasible} otherwise.
If option $\omega_i$ turns out to be feasible, then it produces a
known payoff of $x(\omega_i)$ for the principal and a
known payoff $y(\omega_i)$ for the agent;
if it turns out to be infeasible, then it produces a payoff of $0$ for both.
The only way to evaluate the feasibility of option $\omega_i$ is to pay
a cost of $c_i$ to investigate it.

The principal delegates to the agent the task of proposing a feasible
option, which the principal can either accept or reject.
The principal will not be able to see which options the agent
decides to pay to evaluate as part of this task, but again the principal
can specify a predicate defining the eligible subset of $\omega_i$
that they will accept.
Subject to this constraint, the agent then must decide how to
evaluate options in a way that maximizes its own benefit $y(\omega_i)$
from the option $\omega_i$ it proposes, minus the evaluation cost.
Here too we evaluate the principal's payoff relative to the
scenario in which they performed the evaluation of options themself.
We also consider an extension of this model in which
there is a budget of $n < m$ on the number of options that
the agent can evaluate --- a constraint analogous to the bound
on the number of samples the agent can evaluate in our
first distributional model.

\subsection{Overview of Results}

We begin by showing that for an arbitrary instance of the
distributional model, there is a mechanism the principal can
specify to the agent so that the principal's expected payoff from
delegation is within a factor of $2$ of the expected payoff
they'd receive were they able to search for the solution by themself.
(If the principal were to search by themself, they would examine
$n$ candidate solutions and choose the one that was best for them.)
This mechanism only requires knowledge of the principal's
$x(\omega_i)$ values, not the agent's $y(\omega_i)$ values,
and it has a very simple structure:
depending on the distribution of values, it can be written
as a threshold rule with either a weak threshold,
in which the principal only accepts
proposals $\omega$ for which $x(\omega) \geq \theta$ for some $\theta$,
or a strict threshold,
in which the principal only accepts
proposals $\omega$ for which $x(\omega) > \theta$ for some $\theta$.
In the case when $x(\omega)$ and $y(\omega)$ are distributed independently
with no point masses, the factor of
$2$ in this bound can be improved to $e / (e-1)$.

%% \footnote{We can also
%% obtain the bound of $e / (e-1)$ for distributions with point masses,
%% provided the principal's mechanism is allowed to randomize:
%% for a given proposal of $x(\omega_i)$, the mechanism specifies a probability
%% that the principal will accept the proposal.}

There are several things worth remarking on about this result.
First, the fact that arbitrary instances of the problem have mechanisms
providing provable guarantees of this form suggests a qualitative
argument for the robustness of delegation: no matter how misaligned
the agent's interests are, the principal can ensure an absolute
bound on how much is lost in the quality of the solution.
Second, the mechanism that achieves this bound is very simple and detail-free,
consisting of just a (weak or strict)
threshold on the quality of the solution for the principal.
And third, the mechanism requires knowledge of $n$ (the number of
samples drawn by the agents) but not the values $y(\omega)$.
In this sense, it suggests that it is more important for the
principal to know how much effort the agent has spent on the search
(via $n$) than to know how good the proposed solution is for
the agent (via $y(\omega)$).

% \bkcomment{I'm inserting a reminder here that it would be nice if we can find time to insert a proof that knowledge of $n$ is necessary, somewhere in the paper.}

\xhdr{A connection to prophet inequalities}
These results on threshold mechanisms and their guarantees
follow from a general result at the heart of our analysis ---
a close connection between bounds for delegated search
and {\em prophet inequalities}.
Prophet inequalities are guarantees for the
following type of decision under uncertainty: we see a sequence
of values $s_1, s_2, \ldots, s_n$ in order, with $s_i$ drawn from a
distribution $S_i$, and when we see the value $s_i$ we must
irrevocably decide whether to stop and accept $s_i$, or continue
(in the hope of finding a better value in the future).
Research on prophet inequalities has established the non-trivial
fact that it is possible to design rules whose expected payoff
comes within a constant fraction of the maximum achievable by
a decision-maker who could see all the
realized values $s_1, s_2, \ldots, s_n$ in advance.

Prophet inequalities tend to be established by designing
carefully constructed {\em threshold rules}, in which
the decision-maker accepts $s_i$ if and only if
$s_i$ (weakly or strictly) exceeds a specified threshold $\theta_i$
that can depend on the position $i$.
The key component of our analysis is to establish a close,
though subtle, technical connection between
delegated search and prophet inequalities: roughly speaking,
the sequence of values
$x(\omega_1), x(\omega_2), \ldots, x(\omega_n)$
sampled by
the agent from the set of possible solutions $\Omega$
plays the role of the process
generating $s_1, s_2, \ldots, s_n$;
and the principal and the agent jointly --- through the
principal's specification of the threshold and the agent's
incentive to obey it --- play the role of the decision-maker
who uses a threshold rule for deciding when to stop.
Again, the notion of ``stopping'' here is a bit oblique,
since the principal never sees the full sequence
$x(\omega_1), x(\omega_2), \ldots, x(\omega_n)$
that the agent generates;
this is the sense in which the stopping rule is
jointly constructed by the behavior of the principal and the agent together.

\xhdr{Stronger bounds for independent values}
Using this connection to threshold rules for prophet inequalities,
we can design a much more powerful policy for the principal in
the case when the values of
$x(\omega)$ and $y(\omega)$ on a draw $\omega$ from $\Omega$
are distributed independently,
and when the principal can see both
$x(\omega_i)$ and $y(\omega_i)$ (rather than only $x(\omega_i)$)
in the solution $\omega_i$ proposed by the agent.

To do this, we begin with a stopping rule from the
prophet inequality literature achieving an expected
payoff that is at least $0.745$ times the optimum when the
distributions of the $s_i$ values are independent and
identically distributed \citep{abolhassani-prophet-ineq,correa-prophet-ineq,hill-kertz,kertz86}.
This stopping rule uses a sequence of thresholds $\theta_i$
that decrease with $i$, making the decision-maker naturally
more prone to stop and accept a value as the end of the sequence nears ---
effectively following the idea that one should only accept
a value early if it's very good.

In the context of delegated search when the principal
can observe both $x(\omega_i)$ and $y(\omega_i)$
for a proposed solution $\omega_i$, a related
concept is useful for designing mechanisms: the principal should
only accept a solution $\omega_i$ with $y(\omega_i)$ very large
if $x(\omega_i)$ is very large as well.
The analogy between requiring strong incentive to
accept a value with large $y(\omega_i)$ in delegated search and requiring
strong incentive to accept a value early in the sequence
in the prophet inequality
context can be made precise, and it reveals that the $y(\omega_i)$
values (over the set of candidate solutions considered by the agent)
can be used as a kind of ``continuous time'' parameter
for deriving a threshold: if we think of the candidate solution $\omega_i$
as arriving at continuous time $y(\omega_i)$, then we can derive a
threshold function $\theta(\cdot)$ in which the principal
only accepts $\omega_i$ if
$x(\omega_i)$ (weakly or strictly)
exceeds $\theta(y(\omega_i))$.\footnote{We observe that this
continuous time defined by the $y(\omega_i)$ runs ``in reverse,''
in the sense that large values of $y(\omega_i)$, like small values of time,
place stricter demands on the $x(\omega_i)$ values that can be accepted.}
In this sense, the principal and agent again jointly construct the
stopping rule, with the agent's payoff providing a type of synthetic
temporal ordering that is useful in formulating a threshold policy.

\xhdr{Bounds for the binary model}
We also use the connection to prophet inequalities to derive
bounds for the binary model, where the agent pays to evaluate
the feasibility of pairs from a known list of options
$\omega_1, \omega_2, \ldots, \omega_m$.
Here too the principal can
% define a single threshold $\theta$
designate a predefined eligible set of proposals
so that the mechanism that accepts any eligible proposed $\omega_i$
% with $x(\omega_i) \geq \theta$
yields an expected payoff that is within a factor of two of
the benchmark in which the principal performs the search on their own.
However, the eligibility criterion in this case
is subtler: it depends not only on the principal's
assessment of the proposal's quality, $x(\omega_i)$,
but also on the cost $c_i$ and the {\em a priori}
probability of feasibility, $p_i$.

To establish this bound, we draw on both prophet inequality bounds
and on work of \citet{kleinberg-ec2016}
for the {\em box problem} \citep{weitzman-box-problem};
by considering an ordering of the options by the notion of
{\em reservation price} (or, equivalently,
{\em strike price}) defined in those works, we can
establish a provable guarantee that correctly handles not
only the payoff arising from the $x(\omega_i)$ and $y(\omega_i)$ values
but also the cost incurred by the agent
in evaluating the feasibility of options.

Finally, we derive similar bounds in the more general case where
the agent also has a budget of $n < m$ on the number of options
they can evaluate.
The approach using reservation prices does not directly extend
to this case, but we show that by combining the approach
of \citet{kleinberg-ec2016} with
bounds for stochastic optimization due to
\citet{asadpour-stochastic-submod}, we can obtain more general
bounds for a budgeted variant of the box problem that contains
the case we need for our delegated search guarantee.

We note that it would be a natural open question to consider a
variant of the problem combining characteristics of the two
main versions we consider here: as in the distributional model,
the agent performs independent
% and identically distributed
draws from a space $\Omega$;
but as in the binary model, the agent does not have a fixed bound $n$
on the number of allowed draws, instead incurring a cost to perform
each draw that must be traded off against the eventual payoff
from the sample selected.

\subsection{Further Discussion of Related Work on Delegation}

The theory of delegation in the economics literature is often
viewed as beginning with Bengt Holmstrom's Ph.D. thesis
\citep{holmstrom-thesis,holmstrom-theory-delegation};
this work articulates the basic tension that we see in these
models, between allowing an agent to optimize in a large space
and restricting the agent's freedom of action to prevent them
from pursuing their own objectives too aggressively.
Holmstrom's model considered delegating an optimization problem
over an interval, and a sequence of subsequent papers analyzed the case
in which the agent optimizes over a continuum
\citep{alonso-optimal-delegation,melumad-delegation-no-transfer}.
\citet{armstrong-delegated} propose a model that is very
close to what we consider here, where the optimization takes
place over a discrete set that the agent samples from an
underlying distribution.
By way of comparison between our work and that of
\citet{armstrong-delegated}, we noted the key contrast
earlier in this section:
their paper is largely devoted to identifying cases of the delegated
search problem for which the structure of the optimal mechanism can be
identified, whereas we focus on bounding the inefficiency of
delegated search relative to a benchmark in which the principal
performs the task themself.
It is through our emphasis on these types of bounds that we develop
the connection to the analysis of prophet inequalities.

A distinct line of work in delegation relaxes the constraint that
the principal may only allow or forbid each proposed solution,
and instead allows the principal to add arbitrary amounts of
cost to certain subsets of proposed solutions
\citep{athey-delegation-rigidity,amador-optimal-delegation,ambrus-delegation-money-burning}.
One of the key motivations for such a condition is to model
the strategic role of bureaucracy within an organization:
if management wants to dissuade units within the organization
from proposing certain types of solutions, they can use
bureaucratic measures (requiring more extensive justifications
and processes) that make these solutions selectively more costly
without explicitly forbidding them, and without engaging in
explicit monetary transfers.
\citet{ambrus-delegation-money-burning} propose a model in which
such selective cost increases are in fact part of the optimal
delegation scheme.

Finally, a recent working paper by \cite{taggart} studies
algorithmic delegation in a much more general setting, in which
a principal must choose an action and she delegates this choice
to an agent who is informed of the state of the world. The
principal's and agent's
preferences over actions in every state of the world may differ,
leading to a problem of designating a set of eligible actions among which the agent may choose,
so as to maximize the principal's utility when the agent chooses selfishly among
the eligible actions. This problem is shown
 to be computationally hard in general, but a simple threshold
 policy is shown to achieve a 2-approximation to the optimal
 mechanism under a ``negative bias'' condition. Such a result is
 in the same spirit as our 2-approximation result for threshold
 rules, though the hypothesis under which their result holds, and
 the method of proof, are quite different.

\section{Model and Preliminaries}
\label{sec:prelim}

We begin by making the precise the way in which the
principal and the agent interact, resulting in the
principal's selection of (at most) one element from a set
$\Omega$ of potential solutions.
There are functions $x : \Omega \to \reals$ and $y: \Omega \to \reals$
such that if $\omega \in \Omega$ is selected, then the
principal's utility is $x(\omega)$ and the agent's
utility is $y(\omega)$.
To formalize the possibility
that the principal selects no solution (i.e., perpetuating
the status quo) we identify this
possiblity with a special  ``null outcome'', denoted
$\bot$, and we extend the utility functions from
$\Omega$ to $\Omega_+ \defeq \Omega \cup \{\bot\}$
by setting $x(\bot) = y(\bot) = 0$.

The set $\Omega$ is a probability
space, with probability measure $\mu$, and
the agent has the power to draw independent samples from
$\Omega$ according to $\mu$.
The principal, on the other hand,
can neither draw samples
from $\Omega$ nor directly observe the
outcome of the agent's sampling;
she must rely on her interaction with the agent to
arrive at a selected element of $\Omega$.\footnote{Throughout
this paper we use
feminine pronouns for the principal and masculine
pronouns for the agents.}

Before formalizing our model of interaction,
it is useful to first note some of the
ways in which our basic model can be generalized or specialized.
\begin{itemize}
\item We will initially consider the case of a single
probability measure $\mu$ on $\Omega$, but it is also useful
to consider cases in which there are
multiple probability measures $\mu_1, \mu_2, \ldots, \mu_m$ on $\Omega$,
and the agent has the power to draw independent samples from
any of these distributions.
\item We will generally assume there is a {\em sampling budget} of
$n$ on the number of samples that the agent can draw.
In some of our models, we will also introduce a {\em sampling cost}
$c \geq 0$ for each draw by the agent --- or in the case of
multiple probability measures, a cost $c_i \geq 0$ for sampling from $\mu_i$.
\item We consider both the {\em full-information case} --- in which
the principal knows both the functions
$x : \Omega \to \reals$ and $y: \Omega \to \reals$, and hence
can evaluate the utility of a solution $\omega$ to both herself
and to the agent --- and the {\em limited-information case},
in which the principal only knows her own utility function $x$.
\item The functions $x : \Omega \to \reals$ and $y: \Omega \to \reals$
define random variables on $\Omega$, and we consider both the case
in which they can be arbitrary non-negative functions, and the
case of {\em independent utilities}, when they
are independent random variables.
\end{itemize}

In a later section, we will specialize the formalism to the
{\em binary model} discussed in Section \ref{intro}, in which
each distribution $\mu_i$ is supported on a two-element
set $\{\omega_{0i},\omega_{1i}\}$ such that
$(x(\omega_{0i}), y(\omega_{0i})) = (0,0)$.
In this case we will let $(x_i,y_i)$ denote the pair
$(x(\omega_{1i}),y(\omega_{1i}))$. The binary model captures
a setting in which the feasibility of the $i^{\text{th}}$ solution is
unknown until the agent investigates it, but
the value of the solution to both parties
(if feasible) is known {\em a priori}.

\subsection{A General Definition of Mechanisms for the Principal and Agent}

Let us now formalize how the principal and the agent interact,
resulting in the principal's selection of a solution.
Thus far, our discussion in Section \ref{intro} has focused on
interactions of a very structured form:
the agent draws a set of samples from $\Omega$;
the agent selects one of these samples to present to the principal;
and the principal accepts or rejects it.
But it would be useful to be able to consider
more general formulations for their allowed interactions,
within which the transmission of a single proposal from the agent
to the principal is a particular special case.

To do this, we begin by defining a {\em mechanism} as follows.
A mechanism $M = (\Sigma,g)$
defines a set of signals, $\Sigma$, that
the agent may send, and an allocation function
$g : \Sigma \to \Omega_+$ that specifies
which solution the principal will choose
given the agent's signal.
In such a mechanism, a {\em strategy} for the
agent is specified by a mapping $\sigma : \Omega^* \to \Sigma$,
where $\Omega^*$ denotes the set of finite sequences over $\Omega$,
such that $\sigma(\omega_1,\ldots,\omega_n)$
represents the signal the agent sends
if he sampled $n$ solutions and observed
the sequence $\omega_1,\ldots,\omega_n$.

Suppose the agent observes sequence
$\omega_1,\ldots,\omega_n$ and sends signal
$\sigma$, resulting in outcome $\omega = g(\sigma)$.
In this case, the principal's and agent's utilities
are $x(\omega)$ and $y(\omega)$, respectively, if
$\omega \in \{\omega_1,\ldots,\omega_n\} \cup \{\bot\}$.
Otherwise the principal's utility is 0 and the agent's
is $-1$. In other words, we assume that if the mechanism results
in the principal selecting a solution that was
never sampled by the agent, that solution cannot
be adopted. Instead the status quo is
preserved and the agent suffers a penalty.
This assumption is consistent with our assumption
that the principal lacks the power to directly
search for solutions herself; she can only adopt
solutions that the agent has discovered.
%% this has the same
%% effect for both parties as the null outcome
%% $\bot$ which preserves the status quo.

In models with costly sampling, the specification
of an agent's strategy must also include a
sequential policy $\pi$ for deciding
which sample (if any) to observe next,
given the set of samples already observed.
The principal's and agent's utilities are both
diminished by the sum of costs $c_i$ for the
samples $i$ that the agent observed when
running policy $\pi$. (We deduct this sum
from the principal's utility because we
think of the cost incurred by the agent in searching
for a solution as a kind of ``waste'' that the
principal views as detracting from the overall utility.)

Under our definition of mechanisms, the sequence of
solutions sampled by the agent leads to a signal
(via the agent's strategy $\sigma$), and this
signal leads to a solution in $\Omega_+$ (via
the principal's allocation function $g$).
Composing these two functions, we get a mapping
from the agent's sampled solutions to a single solution:

\begin{defn}[interim allocation function]
\label{interim-alloc}
  If $M$ is a mechanism and $\sigma$ is an
  agent's strategy, the interim allocation
  function of the pair $(M,\sigma)$ is the
  mapping $\ialloc{M,\sigma} : \Omega^* \to \Omega_+$ obtained
  by composing the strategy $\sigma$ with
  the allocation function $g$. In other words,
  $\ialloc{M,\sigma}(\omega_1,\ldots,\omega_n)$
  is the outcome resulting from mechanism $M$,
  when the agent draws sample sequence
  $(\omega_1,\ldots,\omega_n)$ and plays
  according to $\sigma$.
\end{defn}

\subsection{Single Proposal Mechanisms}
\label{sec:single-proposal}

We now show that there is a sense in which it is without
loss of generality
to focus on interactions in which the agent proposes
a single solution from among the ones they sampled,
and the principal either accepts or rejects it.
To do this, we define a simple type of
mechanism called a {\em single proposal
mechanism}, and we show in \Cref{spm}
below that any other mechanism
can be simulated by a single proposal
mechanism.

\begin{defn}
\label{def:spm}
  A {\em single-proposal mechanism} with
  eligible set $R \subseteq \Omega$ is a
  mechanism in which the agent proposes
  one outcome, and the mechanism accepts
  this proposal if and only if it belongs
  to $R$. More formally, $M = (\Sigma,g)$
  is a single-proposal mechanism with
  eligible set $R$ if $\Sigma = \Omega_+$
  and $g$ restricts to the identity
  function on $R$ and the constant
  function $\omega \mapsto \bot$ on
  $\Omega_+ \setminus R$.
\end{defn}

\begin{lem} \label{spm}
  If $M$ is any mechanism and $\sigma$ is
  any strategy constituting a best response to $M$,
  then there exists a single proposal mechanism
  $M'$ and a best response $\sigma'$ to $M'$,
  such that the interim allocation functions
  $\ialloc{M,\sigma}$ and $\ialloc{M',\sigma'}$
  are identical.
\end{lem}
\begin{proof}
  Let $R$ be the range of the interim allocation
  function $\ialloc{M,\sigma}$, i.e.\ the set of
  all possible outcomes of $M$, other than $\bot$,
  when the agent acts according to $\sigma$. Define
  $M'$ to be the single-proposal mechanism with
  eligible set $R$. Let $\sigma'$ be the strategy
  in which the agent observes his tuple of samples,
  $\bm{\omega}=(\omega_1,\ldots,\omega_n)$, and chooses
  strategy $\sigma'(\bm{\omega})=g(\sigma(\bm{\omega}))$. By
  construction the interim allocation functions
  $\ialloc{M,\sigma}$ and $\ialloc{M',\sigma'}$
  are identical. To prove that $\sigma'$ is a
  best response to $M'$, consider any
  $\bm{\omega} \in \Omega^n$ and any
  $\nu \neq \sigma'(\bm{\omega})$.
  Let $y_0 = y(g'(\sigma'(\bm{\omega})))$
  denote the agent's utility when playing
  according to $\sigma'$; note that $y_0 \ge 0$.
  We wish to show
  that the agent cannot benefit by playing $\nu$
  instead, i.e.
\begin{equation} \label{eq:spm.1}
  y(g'(\nu)) \le y_0.
\end{equation}
  If $\nu \not\in R$ then $g'(\nu) = \bot$
  which implies~\eqref{eq:spm.1} since $y_0 \ge 0$.
  If $\nu \in R$ then $\nu = g(\tilde{\sigma})$
  for some $\tilde{\sigma} \in \Sigma$.
  Now~\eqref{eq:spm.1} follows because
  strategy $\sigma$ is a best response for mechanism $M$,
  and $y(g'(\nu))$ denotes the agent's utility
  when playing strategy $\tilde{\sigma}$ in $M$
  whereas $y_0$ denotes his utility when playing
  $\sigma$.
\end{proof}

\section{Analyzing Delegated Search Via Prophet Inequalities}

In this section we develop a formal link between
delegated search mechanisms and prophet inequalities.
It turns out that the relevant prophet inequalities
involve random variables arriving at discrete points
in continuous time, rather than the usual assumption
that they arrive at time points $1,2,\ldots,n$.
Accordingly, we will begin by explaining
the formal model of continuous-time prophet
inequalities in \Cref{ctpe} below. Then
in \Cref{prophet-reduction}
we explain the reduction from delegated
search (in the distributional model) to
continuous-time prophet inequalities.

\subsection{Continuous-time prophet inequalities}
\label{ctpe}

In this section we will be concerned with problems
which involve designing a selection rule to choose
(at most) one element from a random finite set
of pairs $(x_i,t_i) \in \reals_+ \times [0,1]$,
with the goal of maximizing the expected $x$-coordinate
of the chosen element. The $t$-coordinate is
thought of as a time coordinate, and we will
generally (but not exclusively) be concerned
with selection rules that make their choice
without looking into the future, as is ordinarily
the case in the analysis of prophet inequalities.

\begin{defn}[selection rules]
\label{ctsp}
  A \emph{selection rule} is a function $\rho$ from finite
  subsets of $\reals_+ \times [0,1]$ to the set
  $(\reals_+ \times [0,1]) \cup \{\bot\}$, with
  the property that $\rho(S) \in S \cup \{\bot\}$
  for every $S$.
  %% For example, the \emph{optimal selection rule},
  %% denoted by $\ast$, selects from every non-empty set $S$
  %% an element $\ast(S)$ with maximum $x$-coordinate.

  A \emph{stopping rule} is a selection rule that
  chooses element $(x,t)$ from set $S$ without
  looking at the set of elements whose time coordinate
  is greater than $t$. Formally, $\rho$ is a stopping
  rule if it satisfies the following property: for any
  $(x,t) \in \reals_+ \times [0,1]$ and any two sets $S,S'$
  such that $S \cap (\reals_+ \times [0,t]) = S' \cap
  (\reals_+ \times [0,t])$, we have
  $$ \rho(S) = (x,t) \; \Longleftrightarrow \; \rho(S') = (x,t). $$

  An \emph{oblivious stopping rule} with eligible set
  $Q \subseteq \reals_+ \times [0,1]$
  is a stopping rule $\rho_Q$ such that for every $S$,
  $\rho_Q(S)$ is an earliest element of $S \cap Q$
  (i.e., an element of that set with minimum $t$ coordinate)
  or $\rho_Q(S) = \bot$ if $S \cap Q$ is empty.

  A \emph{threshold stopping rule} with threshold $\theta$
  is an oblivious
  stopping rule whose eligible set is of the form
  $(\theta,\infty) \times [0,1]$ or
  $[\theta,\infty) \times [0,1]$.
\end{defn}

\begin{defn}[CTSPs and prophet inequalities]
  A \emph{continuous-time selection problem} (CTSP) is
  an ordered pair $({\mathscr D, R})$
  where  ${\mathscr D}$ is a set of probability distributions
  over finite subsets of $\reals_+ \times [0,1]$,
  and  ${\mathscr R}$ is a set of selection rules.

  A CTSP $({\mathscr D,R})$ satisfies a \emph{prophet inequality
  with factor $\alpha$} if it is the case that for every
  $\D \in {\mathscr D}$ there exists some
  $\rho \in {\mathscr R}$ such that
  $$
     \expect_{S \sim \D}[X_{\rho(S)}] \geq \alpha \cdot
     \expect_{S \sim \D}[X_{\ast(S)}] .
  $$
  Here the random variable $X_{\rho(S)}$
  is defined by specifying that if $\rho(S) = (x,t)$
  then $X_{\rho(S)} = x$, and if $\rho(S) = \bot$ then
  $X_{\rho(S)} = 0$.
  The random variable $X_{\ast(S)}$ is defined
  to be $\max \{ x | (x,t) \in S \}$.
\end{defn}

We now present the prophet inequalities we will use
in this work. To state them,
we define the following families of stopping rules
and distributions on subsets of $\reals_+ \times [0,1]$.
\begin{itemize}
\item ${\mathscr R}_{\text{obliv}}$ is the family
  of oblivious stopping rules.
\item ${\mathscr R}_{\text{thresh}}$ is the family
  of threshold stopping rules.
\item ${\mathscr D}_{\text{ind}}$ is the family
  of random sets whose elements are
  obtained by sampling independently from $n$
  joint distributions. In other words,
  a distribution $\D \in {\mathscr D}_{\text{ind}}$
  is specified by giving a positive number $n$,
  a tuple of joint distributions $\D_1,\ldots,\D_n$
  on $\reals_+ \times [0,1]$, and defining $\D$ to
  be the distribution on $n$-element sets obtained
  by drawing one sample independently from each of
  $\D_1,\ldots,\D_n$.
\item ${\mathscr D}_{\text{iid},n}$ is the family
  of random sets whose $n$ elements are i.i.d.\ samples
  from an atomless distribution with $x$ and $t$ independent,
  i.e.\ a distribution on $\reals_+ \times [0,1]$
  which is a product of atomless distributions.
\item ${\mathscr D}_{\text{iid}}$ is the union of
  ${\mathscr D}_{\text{iid},n}$ over all $n \ge 1$.
%% \item ${\mathscr D}_{\text{iid,Pois}}$ is the family
%%   of random sets with a Poisson-distributed number
%%   of elements, each sampled at random from a common
%%   distribution on $\reals_+ \times [0,1]$ which is
%%   a product of atomless distributions.
\end{itemize}

Our first prophet inequality is
\citeauthor{samuel-cahn}'s
(\citeyear{samuel-cahn})
famous prophet inequality for
threshold stopping rules.

\begin{thm} (\citet{samuel-cahn}) \label{1/2}
  There is a prophet inequality
  with factor $\frac12$ for
  $({\mathscr D}_{\text{ind}}, \, {\mathscr R}_{\text{thresh}})$.
\end{thm}

The second is an improved
prophet inequality for threshold stopping rules
when samples are drawn i.i.d.\ from atomless product
distributions; it can be derived as a corollary of
either \cite[Theorem 19]{ehsani-soda18} or
\cite[Corollary 2.2]{correa-prophet-ineq}.

\begin{thm} (\citet{correa-prophet-ineq,ehsani-soda18}) \label{1-1/e}
  There is a prophet inequality
  with factor $1 - \frac1e$ for
  $({\mathscr D}_{\text{iid}}, \, {\mathscr R}_{\text{thresh}})$.
\end{thm}

Our third prophet inequality again pertains
to the case when samples are drawn
i.i.d.\ from atomless product distributions,
but it allows for oblivious stopping rules
rather than threshold stopping rules.
The discrete-time counterpart to this
prophet inequality can be found in
\cite{hill-kertz,kertz86,correa-prophet-ineq}.

\begin{thm} \label{0.745}
  Let $\alpha = 0.745\dots$ be the solution
  to $\int_{0}^{1} \left( y  - y \ln y + 1 - \frac{1}{\alpha} \right)^{-1}
         \, dy = 1$.
  There is a prophet inequality with factor
  $(1 - \frac6n) \alpha$
  for $({\mathscr D}_{\text{iid},n}, \, {\mathscr R}_{\text{obliv}})$.
\end{thm}

Since the distincton between discrete time and continuous
time is immaterial from the standpoint of analyzing
threshold stopping rules, the first two of these theorems
are equivalent to the existing results for discrete-time
prophet inequalities that we have cited before the theorem
statements. On the other hand, because oblivious stopping
rules are less powerful than general stopping rules,
the third theorem is not an immediate consequence of
the corresponding discrete-time prophet inequality.
Proofs of all three theorems are included in
\Cref{appdx-prophet}; since \Cref{0.745} is the only
novel result among the three theorems, the proofs of the
other two are included only for the purpose of making
our paper self-contained.

To complete this section, we will describe the
stopping rules which achieve the bounds stated
in the three prophet inequalities above.

When the points $\{(x_i,t_i)\}_{i=1}^n$ are independent
but not necessarily identically distributed, choose
threshold $\theta_{1/2}$ to be the median of the
distribution of $\max \{x_i\}$. In other words,
$\theta_{1/2}$ is defined such that the events
$\max \{x_1,\ldots,x_n\} > \theta_{1/2}$ and
$\max \{x_1,\ldots,x_n\} < \theta_{1/2}$ both
have probability at most $\frac12$. Consider the threshold
stopping rule that selects the first pair $(x_i,t_i)$
with $x_i > \theta_{1/2}$, and consider the one whose
selection criterion is $x_i \geq \theta_{1/2}$. The proof
of \Cref{1/2} shows that at least one of these two
stopping rules fulfills a prophet inequality with
factor 1/2.

When the points $\{(x_i,t_i)\}_{i=1}^n$ are
i.i.d.\ and the distributions of $x_i$ and $t_i$ are
atomless and independent, with cumulative distribution
functions $F$ and $G$, respectively,
choose threshold $\theta_{1-1/e}$ such that
$\Pr(\max \{x_1,\ldots,x_n\} > \theta_{1-1/e}) = 1-\frac1e$.
The proof of \Cref{1-1/e} shows that the threshold
stopping rule that selects the first pair $(x_i,t_i)$
such that $x_i > \theta_{1-1/e}$ fulfills a prophet
inequality with factor $1-\frac1e$. Now let
$\beta_n$ be the solution to
$$
   \int_0^n \left( 1 + z +
    \beta_n e^z \right)^{-1} \, dz
 = 1 , $$
and let $z_n(s)$ be the solution of the differential
equation
$$
   \frac{d z}{ds} = 1 + z + \beta_n e^z
$$
with initial condition $z(0)=0$.
%% When the points $\{(x_i,t_i)\}_{i=1}^n$ are
%% i.i.d.\ and the distributions of $x_i$ and $t_i$ are
%% atomless and independent, let $F_x, F_t$ denote the
%% cumulative distribution functions of $x_i$ and $t_i$,
%% respectively.
%% If the number of random samples
%% is a Poisson random variable with expectation $\lambda$,
The oblivious stopping rule that accepts
the first $(x_i,t_i)$ such that
$$ 1 - F(x_i) < z(G(t_i))/n  $$
fulfills a prophet inequality
with factor $(1 - \frac6n) \alpha$.

\subsection{Reducing delegated search to prophet inequalities}
\label{prophet-reduction}

Although delegated search problems and prophet inequalities
appear unrelated at first glance, the tight technical
connection between them is explained by an observation
which is extremely natural in hindsight.
Consider a change of variables that maps the agent's
utility $y$ to a point $t(y) \in [0,1]$, where the
function $t$ is monotonically decreasing. In a
single proposal mechanism with eligible set $R$,
the agent submits the eligible proposal $(x,y)$ with the
highest $y$ value. Similarly, an oblivious stopping
rule with eligible set $Q$ selects the earliest
eligible point $(x,t)$. Since the change of variables
$t(y)$ is monotonically decreasing, the two selection
criteria are equivalent! Thus, designing single proposal
mechanisms that yield high utility for the principal
is equivalent to designing oblivious stopping rules
that yield a high expected value.

In more detail, let $t$ be any continuous, monotonically
decreasing bijection from $[0,\infty)$ to $(0,1]$, for
example $t(y) = e^{-y}$. Under the mapping
$H \,:\, \Omega \to (\reals_+ \times [0,1])$
defined by $H(\omega) = (x(\omega), t(y(\omega)))$,
any distribution on sets of solutions $\{\omega_1,\ldots,\omega_n\}$
induces a distribution $\D$ on sets of pairs
$\{(x_1,t_1),\ldots,(x_n,t_n)\}$. In particular,
our distributional model in which the agent draws
$n$ i.i.d.\ samples from $\Omega$ is mapped, under
this correspondence, to a member of the family of
distributions ${\mathscr D}_{\text{iid},n}$.

There is also a reverse correspondence from
oblivious stopping rules to single proposal
mechanisms and their interim allocation functions.
The oblivious stopping rule $\rho_Q$ with
eligible set $Q$ corresponds to the single
proposal mechanism with eligible set $H^{-1}(Q)$.
More precisely, if $R = H^{-1}(Q)$ and $\sigma$ is
a best response to the single proposal mechanism $M$
with eligible set $R$, then for any sequence of
samples $\bm{\omega} = (\omega_1,\ldots,\omega_n)$,
we have
$$ \rho_Q ( H ( \bm{\omega} ) ) = H ( \ialloc{M,\sigma}(\bm{\omega}) ) .$$
In other words, suppose we run the mechanism $M$; the agent
draws a sequence of samples; and we let the agent choose the best one
(for the agent) that belongs to $R$.
This procedure is equivalent to
running the oblivious stopping rule $\rho_Q$ on the
sequence obtained by transforming all of the agents'
samples to points $(x_i, t_i) = (x_i, t(y_i) )$,
and selecting the earliest such point (ordered by $t_i$)
that belongs to $Q$. Under this correspondence,
threshold stopping rules correspond to single proposal
mechanisms in which a solution is deemed eligible if
the principal's utility exceeds a specified threshold.
Note that this subset of single proposal
mechanisms can be implemented even
when the agent's utility is unobservable.

Combining these observations with \Cref{1/2,1-1/e,0.745},
we obtain the following theorem.

\begin{thm} \label{distributional bounds}
In the distributional model, suppose the agent
draws $n$ i.i.d.\ samples,
and let $x_*$ denote the utility the principal would
attain if she could directly choose her favorite
among these $n$ samples.
\begin{enumerate}
\item There is always a set $X$ of the form
  $(\theta,\infty)$ or $[\theta,\infty)$ such that
  a single proposal mechanism with eligible set
  $\{ \omega \mid x(\omega) \in X \}$ ensures that
  the principal's expected utility is at least
  $\frac12 \expect[x_*].$
\item If the principal and agent have independent
  utilities, each drawn from an atomless distribution,
  then a single proposal mechanism that accepts any
  proposal satisfying $x(\omega) > \theta$, for a
  suitable choice of $\theta$, ensures
  that the principal's expected utility
  is at least $\left( 1 - \frac1e \right) \expect[x_*]$.
\item If the principal and agent have independent
  utilities, each drawn from an atomless distribution,
  and the principal can observe the agent's utility,
  then a single proposal mechanism that accepts
  any proposal satisfying $x(\omega) > \theta(y(\omega))$,
  for a suitable choice of the function $\theta(\cdot)$,
  ensures
  that the principal's expected utility
  is at least $\left( 1 - \frac6n \right) \alpha
  \expect[x_*]$, where $\alpha$ is the constant defined
  in \Cref{0.745}.
\end{enumerate}
\end{thm}
In \Cref{appdx-tight}
we show that the bounds in all three parts of the theorem are
tight with respect to the assumptions made in their respective statements.

\section{Binary outcomes}
\label{binary}

Recall the binary model from \Cref{intro}:
The potential solutions come from a large discrete
set $\Omega = \{\omega_1,\omega_2,\ldots,\omega_m\}$
and the agent's role is to explore which of these options
are feasible to implement. If $\omega_i$ is feasible,
it yields utility $x_i$ for the principal and $y_i$
for the agent --- where the pair $(x_i,y_i)$ is commonly
known to both parties --- and if $\omega_i$ is infeasible
it yields zero utility for both parties.
To explore the feasibility of solution $\omega_i$
the agent must incur a cost of $c_i \ge 0$,
and the probability of success is $p_i > 0$, independently
of the success of other solutions.
These quantities $c_i,p_i$ are again
commonly known to both parties. We will
assume that $c_i \le p_i y_i$ for each
solution $y_i$, since otherwise it is against
the agent's self-interest to explore $\omega_i$,
even if it were assured that the solution would
be adopted if feasible.

\subsection{Optimal search policies: Weitzman's box problem}
\label{weitzman}

If the principal were conducting the search
by herself (without delegation to an agent),
this model would correspond to
a special case of the {\em box problem} introduced
by \citet{weitzman-box-problem}. The optimal
search policy is simple but
surprisingly subtle: it assigns to each
option a priority $z_i$ satisfying
$\expect[ (x_i - z_i)^+ ] = c_i$ --- which
in our case entails setting
$z_i =  x_i - c_i/p_i$ --- and then
explores options in decreasing order
of priority, selecting the first feasible
one in this ordering or stopping when
all remaining unexplored options have
$z_i < 0$.

Now suppose that the principal instead
delegates the search to an agent who
bears the cost of exploration, by running
a single-proposal mechanism with eligible
set $R$. Then the agent faces a different
instance of the box problem, in which
the set of options is limited to $R$, and
the costs and success probabilities of the
options is the same as before, but the
value of option $i$ (if feasible) is $y_i$
rather than $x_i$. This means the agent
prioritizes boxes in decreasing order of
$w_i = y_i - c_i/p_i$ rather than $z_i = x_i - c_i/p_i$,
and recommends the first box in this ordering
that is discovered to be feasible.

To summarize, the delegated search problem in
the binary model is analogous to Weitzman's box
problem, but with the important distinction
that the searcher (the principal) is not allowed
to choose the order in which to open the boxes.
Instead the problem specifies an exogenous ordering
of the boxes --- corresponding to the agent's ranking
of options by decreasing $w_i$ --- and the searcher
is only free to decide which boxes in this sequence
should be opened and which ones should be skipped,
corresponding to the principal's problem of
choosing the set $R$. Since this problem may
be of independent interest, we devote \Cref{exogenous}
below to presenting a solution that always achieves
at least half of the expected value of running the
optimal search procedure that is allowed to inspect
the boxes in any order it desires. Interestingly, the analysis
is based on prophet inequalities, specifically
\Cref{1/2} and its proof. It implies there is an
approximately optimal mechanism with the following structure.
For any half-infinite interval $X$
of the form $X=(\theta,\infty)$ or $X=[\theta,\infty)$,
let $R(X) = \{ \omega_i \mid z_i \in X \}$ and define $M(X)$ to be
the single-proposal mechanism in which a proposal
$\omega_i$ is eligible if it is feasible and belongs to $R(X)$.

\subsection{The Box Problem with an Exogenous Ordering}
\label{box}

In this section we recapitulate some background material
about \citeauthor{weitzman-box-problem}'s (\citeyear{weitzman-box-problem})
box problem. In this problem\footnote{The following description
constitutes a special case of Weitzman's
problem. The general case incorporates geometric time discounting
and time delays.} there are $m$ boxes, each containing
an independent random prize. The prize in box $i$ is denoted $v_i$,
and the cost of opening the box is $c_i$. A searcher may open any
number of boxes sequentially, or may cease the search at any time
and claim a prize from at most one of the open boxes. The problem
is to design an optimal sequential search policy.
Weitzman proves that if each box is assigned
a {\em priority} $z_i$ defined by the equation
$\expect[(v_i-z_i)^+] = c_i$, then the optimal
sequential search policy opens boxes in decreasing
order of priority, stopping at the first time when the
highest prize inside an open box exceeds the highest
priority of a closed box, or at the first time when
the priority of every remaining closed box is negative,
whichever comes sooner.

\citet{kleinberg-ec2016} provided a proof of optimality
of Weitzman's procedure in which the priority $z_i$ is
interpreted as the ``strike price'' of a real option with fair
value $c_i$. An important quantity in their analysis is the
``covered call value'', which is simply the random variable
$\kappa_i = \min \{ v_i, z_i \}$. We restate the following
lemma\footnote{Lemma 1 of the full version of their paper,
\url{http://dx.doi.org/10.2139/ssrn.2753858}.} from their work.

\begin{lem}(\citet{kleinberg-ec2016})  \label{kww}
For any sequential search procedure and any box $i$,
let $A_i,B_i$ be the indicator
random variables of the event that the procedure
selects box $i$ and the event that it opens box $i$,
respectively. The inequality
\begin{equation} \label{eq:kww}
  \expect[ A_i v_i - B_i c_i ] \leq
  \expect[ A_i \kappa_i ]
\end{equation}
is satisfied by every search procedure, and equality holds
if and only if the search procedure is \emph{non-exposed},
meaning that $A_i = B_i$ at every sample point where
$v_i > z_i$.
\end{lem}

\begin{cor} \label{max-kappa}
  For any sequential search procedure, the expected net value
  of running the procedure (i.e., the value of the selected
  box minus the combined cost of opening boxes) is bounded
  above by the expectation of the maximum covered call value,
  i.e.
  \begin{equation} \label{eq:max-kappa}
    \expect \left[ \sum_{i=1}^m A_i v_i - \sum_{i=1}^m B_i c_i \right]
    \leq
    \expect \left[ \sum_{i=1}^m A_i \kappa_i \right].
  \end{equation}
\end{cor}
The corollary is immediate, by summing inequality~\eqref{eq:kww}
over boxes $i=1,\ldots,m$.

Now consider the box problem with an exogenous ordering
of boxes, where the searcher is limited to considering the
boxes one by one in the specified order, and once she
decides to leave a box closed or to leave the prize within
unclaimed, she cannot later return to the box and open it
or claim its prize.
We define a type of policy that we call a \emph{$\kappa$-thresholding policy};
the reason for the name will become apparent in
the subsequent \Cref{kappa-thresh}, which shows that these
policies correspond to a threshold rule applied to the sequence
of covered call values $\kappa_i$.

\begin{defn} A \emph{$\kappa$-thresholding policy} for the
box problem with exogenous ordering is a policy that
operates as follows. There is a half-infinite interval
$X = (\theta,\infty)$ or $X = [\theta,\infty)$ called the
\emph{target interval}. The
policy declines to open any box $i$ with $z_i \not\in X$.
Otherwise, if $z_i \in X$, the policy opens the box and
claims the prize inside if and only if $v_i \in X$.
\end{defn}

\begin{lem} \label{kappa-thresh}
  Every $\kappa$-thresholding policy is non-exposed.
  The expected net value of running a $\kappa$-thresholding
  policy with target interval $X$ is exactly the same
  as the expected value selected by the threshold
  stopping rule that observes the random sequence
  $\kappa_1,\kappa_2,\ldots,\kappa_n$ and selects
  the first element of this sequence that belongs
  to $X$.
\end{lem}
\begin{proof}
  The policy is non-exposed because
  $z_i \not\in X$ implies $A_i = B_i = 0$, while $z_i \in X$
  and $v_i > z_i$ imply $A_i = B_i = 1$.
  Hence the left and right sides of~\eqref{eq:kww}
  are equal for every box, and the net value of
  running the policy is $\expect \left[ \sum_{i=1}^m A_i \kappa_i \right]$,
  i.e.\ the expected covered call value of the
  box the policy selects.
  By design, the policy perfectly simulates
  the threshold stopping rule that chooses
  the first element of the sequence $\kappa_1,\ldots,\kappa_n$
  that belongs to $X$; this is because it selects the first box
  such that $z_i$ and $v_i$ both belong to $X$, which is
  also the first box such that $\kappa_i$ belongs to $X$.
\end{proof}

\begin{thm} \label{exogenous}
  For every instance of the box problem with exogenous
  ordering, there is a $\kappa$-thresholding policy whose
  expected net value is at least half that of Weitzman's
  optimal search procedure (which endogenously
  selects the ordering of the boxes).
\end{thm}
\begin{proof}
  \Cref{kappa-thresh} reduces the analysis of
  $\kappa$-thresholding policies to a question
  about prophet inequalities. In particular,
  the expected net value of running a
  $\kappa$-thresholding policy is equal to the
  expected covered call value of the random
  element selected from the sequence
  $\kappa_1,\ldots,\kappa_n$ by a
  particular threshold stopping rule.
  Since \citeauthor{samuel-cahn}'s
  (\citeyear{samuel-cahn}) prophet inequality
  (\Cref{1/2} above) implies that
  threshold stopping rules can always
  attain at least half the expectation
  of the maximum random variable in the
  sequence, it follows
  that there is a $\kappa$-thresholding policy whose
  expected net value is at least half the
  expectation of the maximum covered call value.
  \Cref{max-kappa} ensures that the latter is an
  upper bound on the expected net value of Weitzman's
  optimal search procedure.
\end{proof}

\subsection{An approximately optimal mechanism}
\label{proof-M(X)}

Recall that for a half-infinite interval
$X=(\theta,\infty)$ or $X = [\theta,\infty)$,
the mechanism $M(X)$ is defined to be the single
proposal mechanism whose eligible set consists
of solutions $\omega_i$ that are feasible and
satisfy $z_i \in X$.

\begin{thm}  \label{M(X)}
  There exists a choice of $X$ such that the expected
  net value of mechanism $M(X)$ --- i.e., the principal's
  value for adopting the agent's proposal, if adopted,
  minus combined cost of all the alternatives explored --- is
  at least half of the expected net value the principal
  could achieve by performing the optimal search herself
  (without delegation).
\end{thm}
\begin{proof}
  Convert the delegated search problem into a box problem
  with exogenous order, where the order is defined by sorting the
  solutions $\omega_1,\ldots,\omega_m$ in non-increasing
  order of the agent's priority value $w_i = y_i - c_i / p_i$,
  and the value $v_i$ inside box $i$ is defined to be
  $x_i$ if $\omega_i$ turns out to be feasible, 0 otherwise.

  According to \Cref{exogenous} there exists a choice of $X$
  such that the $\kappa$-thresholding policy with target set $X$
  attains at least half the expected net value of the optimal
  search procedure. This thresholding policy goes through
  boxes in the given order, i.e.\ descending $w_i$, and
  opens only those with $z_i \in X$, selecting the first
  one such that $v_i \in X$. Note that among the boxes
  which the policy opens, the first one with $v_i \in X$
  is also the first one corresponding to a feasible $\omega_i$.
  This is because an infeasible $\omega_i$ has $v_i = 0$ hence
  $v_i \not\in X$, whereas a feasible $\omega_i$ has $v_i = x_i \ge z_i$,
  hence $v_i \in X$.
%% In other words, the $\kappa$-thresholding
%%   policy with target set $X$ opens the same set of boxes, and
%%   selects the same box, as the agent's best response to
%%   mechanism $M(X)$.

  Recall from \Cref{weitzman} that the agent's best response
  to mechanism $M(X)$ is to go through the elements of $R(X)$ in
  decreasing order of $w_i$, stopping and proposing the first one
  that is discovered to be feasible. This is exactly the behavior
  of the $\kappa$-thresholding policy with target set $X$,
  as derived in the preceding paragraph. Hence the mechanism
  $M(X)$ coupled with the agent's best response behavior
  emulates the $\kappa$-thresholding policy which attains
  at least half the expected net value of the optimal
  search procedure.
\end{proof}

\subsection{Limiting the number of samples}
\label{limit-samples}

In some cases the number of distinct potential solutions,
$m$, may be prohibitively large, and the agent may only
have the power to explore the feasibility of a limited
number of them, $n < m$. In this case, if the principal
were to conduct the search autonomously without delegation --- subject to the
same costs $c_i$ and the same upper bound, $n$, on the
total number of solutions that can be tested for
feasibility --- it may require a very complex procedure.
Nevertheless, we will provide in this section a
simple delegated search mechanism such that it is
easy for the agent to compute a search procedure
that is a best response to the mechanism, and the
outcome of running the mechanism with this best
response attains at least $\frac12 \left( 1 - \frac1e \right)
\approx 0.316$ of the net expected value of
the (potentially complex) optimal procedure.

The key observation is the following lemma, which
provides a useful upper bound on the value of running
the optimal search procedure.

\begin{lem} \label{max-of-n}
  In the box problem with $m > n$ boxes, if the
  searcher is limited to open at most $n$ boxes
  before claiming a prize, then the expected net
  value of any search procedure is bounded above
  by $\expect[ \max_{i \in S} \kappa_i ]$ where
  $S$ is the random set of boxes that the procedure
  opens.
\end{lem}
\begin{proof}
  Sum up the inequality~\eqref{eq:kww} over all
  boxes and note that $A_i = 0$ for $i \not\in S$,
  to derive
  \[
    \expect \left[ \sum_{i=1}^m A_i v_i - \sum_{i=1}^m B_i c_i \right]
    \le
    \expect \left[ \sum_{i \in S} A_i \kappa_i \right].
  \]
  The lemma follows by noting that $\sum_{i \in S} A_i \kappa_i
  \le \max_{i \in S} \kappa_i$ because $\sum_{i \in S} A_i \le 1$.
\end{proof}

\begin{lem} \label{non-adaptive}
  There exists a (non-random) set $T$ of cardinality $n$,
  such that $\expect [ \max_{i \in T} \kappa_i ] \ge
  (1 - \frac1e) \expect[ \max_{i \in S} \kappa_i ]$,
  where $S$ is the random set of solutions explored
  by the optimal search procedure subject to a
  contraint of exploring at most $n$ solutions.
\end{lem}
\begin{proof}
  The problem of adaptively exploring a random set $S$ of at
  most $n$ solutions to maximize $\expect [ \max_{i \in S} \kappa_i ]$
  is a special
  case of the stochastic monotone submodular function maximization
  problem studied by \citet{asadpour-stochastic-submod},
  in which the role of the monotone submodular function
  $f : \reals_+^n \to \reals_+$ is played by the function
  $f(\lambda_1,\ldots,\lambda_n) = \max \{\lambda_i\}$,
  and role of the matroid constraint is played by the
  cardinality constraint that at most $n$ elements may
  be probed. Theorem 1 of \citet{asadpour-stochastic-submod},
  which asserts that the adaptivity gap of stochastic monotone
  submodular maximization is $\frac{e}{e-1}$, specializes in
  the present case to the assertion stated in the lemma.
\end{proof}

\begin{thm} \label{0.316}
  Consider delegated search in the binary model
  with a constraint that no more than $n$ solutions
  can be examined for feasibility. There exists a
  mechanism that attains at least
  $\frac12 \left( 1 - \frac1e \right)$ fraction
  of the expected net value of the optimal search
  procedure subject to the same limitation of
  examining at most $n$ solutions.
\end{thm}
\begin{proof}
  According to \Cref{max-of-n,non-adaptive}, there
  is an $n$-element set $T \subseteq \Omega$ such that
  the optimal search procedure that is limited to
  explore only solutions in $T$ is able to attain
  at least $1 - \frac1e$ fraction of the expected
  net value of the optimal search procedure that is
  limited to examine at most $n$ solutions but can
  (adaptively) choose any $n$ elements of $\Omega$
  during its search. When the set of solutions is
  restricted to $T$, the constraint that at most
  $n$ solutions can be examined becomes irrelevant
  since $T$ only has $n$ elements. Thus,
  \Cref{M(X)} guarantees the existence of a
  delegated search mechanism that is at least
  half as good as the optimal search procedure
  limited to $T$, and is consequently at least
  $\frac12 \left( 1 - \frac1e \right)$ times
  as good as the optimal search procedure
  limited to examine at most $n$ solutions.
Moreover, by applying the algorithm in
\citet{asadpour-stochastic-submod}
used to prove \Cref{non-adaptive}, we can implement this
policy in polynomial time with a loss of a further additive $\eps$
in the approximation ratio,
thus obtaining a bound of $\frac12 \left( 1 - \frac1e - \eps\right)$
efficiently.
\end{proof}

 % \begin{acks}

\xhdr{Acknowledgements.}
   This work was supported in part by NSF grants CCF-1512964,
   CCF-1740822, and SES-1741441, a grant from the
   MacArthur Foundation, and a Simons Investigator Award.
   The authors would like to thank Brendan Lucier,
   Jens Ludwig, Sendhil Mullainathan, Rad Niazadeh,
   and Glen Weyl for helpful discussions, and
   The Nines in Ithaca, NY, for the many deep-dish
   pizzas that were consumed during the course
   of this research.
 % \end{acks}

% Bibliography
\bibliographystyle{apalike}
\bibliography{header,refs}

\appendix

\section{Proofs of Prophet Inequalities}
\label{appdx-prophet}

In this appendix, we provide proofs of the three
main prophet inequalities used in this work,
\Cref{1/2,1-1/e,0.745}.

\subsection{Threshold stopping rules for independent distributions}
\label{samuel-cahn}

In this section we provide a proof of the following
theorem from \Cref{ctpe}. The theorem (stated in a
different form) was originally proven by
\citet{samuel-cahn}; we provide a proof here
for the sake of making the paper self-contained.

\begin{thm}
  There is a prophet inequality
  with factor $\frac12$ for
  $({\mathscr D}_{\text{ind}}, \, {\mathscr R}_{\text{thresh}})$.
\end{thm}
\begin{proof}
  Consider any distribution $\D = \D_1 \times \cdots \times \D_n$
  in ${\mathscr D}_{\text{ind}}$ and let $x_1,\ldots,x_n$ denote
  independent random variables representing the $x$-coordinates
  of random samples from $\D_1,\ldots,\D_n$, respectively.
  Note that the subscripts on the variables $x_1,\ldots,x_n$
  represent the distributions from which they were sampled,
  not necessarily the order in which they arrive, since their
  corresponding time coordinates $t_1, \ldots, t_n$ may
  not be in ascending order. However, this issue will be
  immaterial in the proof because our argument is insensitive
  to the arrival order of $x_1,\ldots,x_n$.

  Let $x_* = \max \{x_1,\ldots,x_n\}$ and
  choose a threshold $\theta$ defining two
  semi-infinite intervals
  $$X_0 = (\theta,\infty), \quad
  X_1 = [\theta,\infty)$$ such that
  $p_0 \stackrel{\Delta}{=} \Pr(x_* \in X_0) \le \frac12$ and
  $p_1 \stackrel{\Delta}{=} \Pr(x_* \in X_1) \ge \frac12$.
  Let $q$ be the solution to the equation $q p_0 + (1-q) p_1 = \frac12$,
  and let $x_\tau$ denote the random variable
  defined by the following sampling process: first choose the interval
  $X=X_0$ with probability $q$ and $X=X_1$ with probability $1-q$.
  Then apply the threshold stopping rule with eligible set
  $X \times [0,1]$ to select an element $(x,t) \in S$ and
  let $x_\tau$ denote the $x$-coordinate of that element.
  We will prove that $\expect x_\tau \ge \frac12 \expect x_*$.
  Since $\expect x_\tau$ is a convex combination of the
  expected value of applying the threshold stopping rule
  with eligible set $X_0 \times [0,1]$ and the one with
  eligible set $X_1 \times [0,1]$, it will follow that the
  better of those two stopping rules fulfills a prophet
  inequality with factor $\frac12$.

  To compare $\expect x_*$ with $\expect x_\tau$ we reason
  as follows. First, we have the following easy upper bound
  on $x_*$.
  \begin{equation} \label{eq:x*}
    \expect x_* \le \theta + \sum_{i=1}^n \expect (x_* - \theta)^+
  \end{equation}
  To put a lower bound on $x_\tau$, let $Y_i$ denote an indicator
  random variable for the event that $\{x_1,\ldots,x_n\} \cap X = \{x_i\}$.
  Note that this event, when it happens, implies that
  $x_\tau = x_i$.
  \begin{align}
\nonumber
    \expect x_\tau & \ge \Pr(x_\tau \in X) \cdot \theta +
    \sum_{i=1}^n \expect[Y_i (x_i - \theta)] \\
\nonumber    &=
    \frac12 \theta +
    \sum_{i=1}^n \Pr(Y_i=1) \expect[ x_i - \theta \given Y_i=1 ] \\
\intertext{because $\Pr(x_\tau \in X) = \frac12$ by our
  construction of the random set $X$}
\nonumber    &=
    \frac12 \theta +
    \sum_{i=1}^n \Pr(\{x_1,\ldots,x_n\} \cap X \subseteq \{x_i\} ) \cdot
      \expect[ (x_i - \theta)^+ ] \\
\nonumber     &\ge
    \frac12 \theta +
    \sum_{i=1}^n \Pr(x_* \not\in X) \cdot \expect[ (x_i - \theta)^+ ] \\
    &=
    \frac12 \left( \theta + \sum_{i=1}^n \expect[ (x_i - \theta)^+ ] \right),
\label{eq:xtau}
\end{align}
where the last line follows because $\Pr(x_* \not\in X) = \frac12$,
again by our construction of the random set $X$. Finally,
the theorem follows by combining~\eqref{eq:x*} with~\eqref{eq:xtau}.
\end{proof}

\subsection{Threshold stopping rules for i.i.d.\ atomless distributions}
\label{thresh-stop-iid}

This section provides a proof of the following lemma,
which is equivalent to \Cref{1-1/e} from
\Cref{ctpe}.

\begin{lem} \label{thresh-stoprule}
  If $X_1,X_2,\ldots,X_n$ is a sequence
  of i.i.d.\ random variables, each sampled
  from an atomless distributions, and $\tau$
  is a  threshold stopping rule
  such that $\Pr(\tau > 1) = \exp(-1/n)$,
  then
  \begin{equation} \label{eq:thresh-stoprule}
     \expect[X_\tau] \geq \left( 1 - \tfrac1e \right)
     \expect \left[ \max_{1 \le t \le n} X_t \right].
  \end{equation}
\end{lem}
\begin{proof}
  Let $X_* = \max_{1 \le t \le n} X_t$. We have
\begin{align}
  \label{eq:thresh-stoprule.1}
  \expect[X_\tau] &= \int_{0}^{\infty} \Pr(X_\tau > y) \, dy \\
  \label{eq:thresh-stoprule.2}
  \expect[X_*] &=
  \int_{0}^{\infty} \Pr(X_* > y) \, dy  .
\end{align}
  To compare the integrands at any specified $y \ge 0$, we
  consider the cases $y < \theta$ and $y > \theta$ separately.
  (The case $y=\theta$ is omitted because it contributes zero to
  both integrals.) When $y < \theta$ the inequality $X_\tau > y$
  is satisfied whenever $\tau \le n$. Since we are
  assuming $X_1,\ldots,X_n$ are identically distributed, and
  the threshold stopping rule has the same behavior at every
  point in time, we have
\[
  \forall t \;\; \Pr(\tau > t \given \tau > t-1) =
  \Pr(\tau > 1) = \exp(-1/n).
\]
  Hence the probability that the stopping rule does
  not stop at any time $t \le n$ is
\[
  \prod_{t=1}^{n} \Pr(\tau > t \given \tau > t-1) =
  \exp(-1/n)^n = \frac1e ,
\]
  and therefore for $y < \theta$,
\[
  \Pr(X_\tau > y) = 1 - \tfrac1e \ge
  \left( 1 - \tfrac1e \right) \cdot \Pr( X_* > y ) .
\]
  Meanwhile, for $y > \theta$, the probability
  that the stopping rule stops at time $t$ and
  selects an element of value greater than $y$ is
  $\exp(-1/n)^{t-1} (1 - F(y))$. Summing over $t$,
  we have
\[
  \Pr(X_\tau > y) =
  \left( \sum_{t=1}^n \exp \left( - \tfrac{t-1}{n} \right) \right)
  \left( 1 - F(y) \right) =
  \left( 1 - \tfrac1e \right)
  \left( 1 - \exp(-1/n) \right)^{-1}
  \left( 1 - F(y) \right).
\]
  Using the fact that $1 - \exp(-1/n) < 1/n$, we find that
\[
  \Pr(X_\tau > y) \ge
  \left( 1 - \tfrac1e \right) n \left( 1 - F(y) \right) \ge
  \left( 1 - \tfrac1e \right) (1 - F(y))^n =
  \left( 1 - \tfrac1e \right) \Pr(X_* > y).
\]
  Since we have proven that
  $\Pr(X_\tau > y) \ge \left(1 - \tfrac1e \right) \Pr(X_* > y)$
  for all $y \ne \theta$, we may integrate over $y$ and
  combine with~\eqref{eq:thresh-stoprule.1}-\eqref{eq:thresh-stoprule.2}
  to obtain~\eqref{eq:thresh-stoprule}.
\end{proof}

\subsection{Oblivious stopping rules applied to i.i.d. atomless
  product distributions}
\label{sec:0.745}

Finally, in this subsection we furnish the proof of
\Cref{0.745} from \Cref{ctpe}. The oblivious stopping
rule that fulfills a prophet inequality with factor
$\alpha - O(\frac{ \log n } { n })$ is more complicated
than the threshold stopping rules analyzed earlier.
It is defined by first solving a differential equation
\begin{equation} \label{z-diffeq}
  \frac{dz}{ds} = 1 + z + \beta_n e^z
\end{equation}
with initial condition $z(0)=0$ to define a function
$z(s)$ mapping $[0,1]$ to $[0,n]$. The constant
$\beta_n$ is chosen so that $z(1) = n$. Since
the differential equation~\eqref{z-diffeq} implies
\[
  s = \int_0^{z(s)} \frac{dz}{1 + z + \beta_n e^z}
\]
for all $s>0$ such that $z(s)$ is defined,
the boundary condition $z(1)=n$ requires the equation
\begin{equation} \label{eq:integral-z}
  1 = \int_0^n \frac{dz}{1 + z + \beta_n e^z}
\end{equation}
to be satisfied, and we treat this equation
as the definition of $\beta_n$.

Denoting the cumulative distribution functions
of $x_i$ and $t_i$ by $F$ and $G$, respectively,
we will be analyzing the oblivious stopping
rule which selects the earliest $(x_i,t_i)$
satisfying
$$ 1 - F(x_i) < z(G(t_i)) / n . $$
We will prove that this rule satisfies
a prophet inequality with a specific
factor $\alpha_n$ to be determined later.
(Equation~\eqref{alphan} below defines $\alpha_n$.)
Let $x_\tau$ denote the
$x$-coordinate of the element $(x_i,t_i)$
selected by the oblivious stopping rule,
and let $x_*$ denote the random variable
$\max_{1 \le i \le n} \{x_i\}$. As in the proof
of \Cref{thresh-stoprule}, we will make use
of the equations
\[
  \expect x_* = \int_0^{\infty} \Pr(x_* > y) \, dy, \quad
  \expect x_{\tau} = \int_0^{\infty} \Pr(x_\tau > y) \, dy
\]
which reduces the task of proving a prophet inequality
with factor $\alpha_n$ to the task of proving
\[
  \forall y \;\;\; \alpha_n \cdot \Pr(x_* > y) \le \Pr(x_\tau > y).
\]
At this point a couple of observations will slightly
simplify the analysis. If we change coordinates to
replace $t_i$ with $G(t_i)$ this has no effect on the
behavior of the stopping rule, and of course if has
no effect on $x_*$, so we are free to adopt this
reparameterization and assume henceforth that $t_i$
is uniformly distributed in $[0,1]$. In particular,
this means the stopping rule simplifies to choosing
the first $(x_i,t_i)$ such that
$ 1 - F(x_i) < z(t_i)/n. $
The next simplification comes from
introducing the variable $q = n \cdot (1 - F(y))$
and writing the
event $x_* > y$ in the form
$1 - F(x_*) < 1 - F(y) = q/n$.
The probability of this event is
$$
  1 - \prod_{i=1}^{n} \Pr \left( 1 - F(x_i) \ge \tfrac{q}{n} \right)
  =
  1 - \left(1 - \frac{q}{n} \right)^n .
$$
The function $\lambda \mapsto \left( 1 + \frac{\lambda}{n} \right)^n$
will be appearing frequently throughout this calculation
so we will assign it a name: $\exp_n(\lambda)$. (This name
is inspired by the fact that $\exp_n(\lambda) \to \exp(\lambda)$
as $n \to \infty$ for any fixed $\lambda$, so $\exp_n$ is a
degree-$n$ polynomial approximation to the exponential function.)

We have derived that for $q = n \cdot (1 - F(y))$,
$$
  \Pr( x_* > y ) = 1 - \exp_n(-q).
$$
Now what about $\Pr( x_\tau > y )$?
We can calculate this probability by
integrating, with respect to $t \in [0,1]$,
the probability that the stopping rule
selects a point in $(y, \infty) \times (t, t+dt)$.
In order for $(x_i,t_i)$ to be this point, the
following things must happen.
\begin{enumerate}
\item $t_i \in (t,t+dt)$. This has probability $dt$.
\item $1 - F(x_i) < z(t_i)/n \approx z(t)/n$.
\item $1 - F(x_i) < 1 - F(y) = q/n$.
\item None of the points $(x_j,t_j)$, for $j \neq i$,
satisfy the stopping condition with $t_j < t_i$.
\end{enumerate}
The conjunction of the first three events has
probability $\frac1n \min \{ z(t),q \} \, dt$.
Let $Z(t) = \int_0^t z(u)$ so that
$Z(t)/n$ is the probability that a specific
point $(x_j,t_j)$ satisfies $t_j < t$ along
with the stopping condition $1 - F(x_j) < z(t_j)/n$.
The fourth event above has probability
$( 1 - Z(t)/n )^{n-1} > \exp_n(-Z(t))$.
Hence
\[
  \Pr( x_\tau > y ) > \int_0^1 \min\{z(t),q\} \exp_{n}( -Z(t) ) \, dt.
\]
%% The integral on the right side can be simplified by
%% writing $q = z(r)$ and using
%% the observations that $z(t) = \frac{d}{dt} Z(t)$
%% and that $\exp_n(u) = \frac{d}{du} \exp_{n+1}(u)$.
%% \begin{align*}
%%  \int_0^1 \min(z(t),z(r)) \exp_{n}( -Z(t) ) \, dt &=
%%  \int_0^r z(t) \exp_{n}( -Z(t) ) \, dt \, + \,
%%  z(r) \int_r^1 \exp_{n}( -Z(t) ) \, dt \\
%%  &= 1 - \exp_{n+1} ( - Z(r) ) + z(r) \int_r^1 \exp_n( - Z(t) ) \, dt .
%% \end{align*}
We are left with proving that for a suitable choice
of $\alpha_n$, the inequality
\begin{equation} \label{expn}
%  \alpha_n \left( 1 - \exp_n( - z(r) ) \right) \le
%  1 - \exp_{n+1} ( - Z(r) ) + z(r) \int_r^1 \exp_n( - Z(t) ) \, dt
   \alpha_n \left( 1 - \exp_n( - q ) \right) \le
   \int_0^1 \min\{z(t),q\} \exp_n(-Z(t)) \, dt
\end{equation}
holds for all $q \in [0,n]$.
At this point a useful identity comes to our aid.
\begin{lem} \label{mysterious-useful}
  The functions $z$ and $Z$ defined above
  satisfy the equation
\begin{equation} \label{exp}
  \frac{1}{1+\beta_n} \left(1 - e^{-q} - q e^{-n}\right)  =
%  1 - \exp( - Z(r) ) + z(r) \int_r^1 \exp( - Z(t) ) \, dt
  \int_0^1 \min\{z(t),q\} \exp(-Z(t)) \, dt
\end{equation}
  for every $q \in [0,n]$.
\end{lem}
\begin{proof}
  The proof is a brief but fairly opaque calculation
  using the differential equation~\eqref{z-diffeq} that
  defines the function $z$. First, observe that $z$
  satisfies the equation
$
  z''(r) = z'(r) \cdot ( z'(r) - z(r) )
$
  because if one differentiates both sides of~\eqref{z-diffeq}
  one obtains
\[
  z''(r) = z'(r) + \beta z'(r) e^{-z(r)}
    = z'(r) \cdot ( 1 + \beta e^{-z(r)} )
    = z'(r) \cdot ( z'(r) - z(r) ).
\]
  Second, observe that the function $h(r) = z'(r) \exp(Z(r)-z(r))$ is
  constant: its derivative satisfies
\begin{align*}
    h'(r) &= z''(r) \exp(Z(r) - z(r))
            + z'(r) ( Z'(r) - z'(r) ) \exp(Z(r)-z(r)) \\
          &= [ z''(r) - z'(r) \cdot (z'(r) - z(r)) ] \exp(Z(r)-z(r)) = 0.
\end{align*}
  Since $h(0) = z'(0) = 1 + \beta_n$, we may conclude that
  $h(r) = 1+\beta_n$ for all $r$, which means
\begin{equation} \label{eq:exp-Z}
  \exp(-Z(r)) = \frac{1}{1+\beta_n} z'(r) \exp(-z(r)) .
\end{equation}
  Integrating both sides of~\eqref{eq:exp-Z} we obtain
\begin{align}
  \int_r^1 \exp(-Z(t)) \, dt &=
  \frac{1}{1+\beta_n} \int_r^1 \exp(-z(t)) z'(t) \, dt
\nonumber \\
  &=
  \frac{1}{1+\beta_n} \left( e^{-z(r)} - e^{-z(1)} \right) =
  \frac{1}{1+\beta_n} \left( e^{-z(r)} - e^{-n} \right).
\label{eq:integral-exp-Z}
\end{align}
Let $r = z^{-1}(q)$. The function $z(t)$ is increasing in $t$,
so $\min\{z(t),q\}$
equals $z(t)$ for $t \le r$, and it equals $q$ for $t > r$.
Hence,
\begin{align*}
  \int_0^1 \min\{z(t),q\} \exp(-Z(t)) \, dt &=
  \int_0^r z(t) \exp( -Z(t) ) \, dt \, + \,
  q \int_r^1 \exp( -Z(t) ) \, dt \\
  &= 1 - \exp ( - Z(r) ) + z(r) \int_r^1 \exp( - Z(t) ) \, dt \\
  &=
  1 - \left( \frac{1}{1+\beta_n} \right) z'(r) e^{-z(r)} +
  \frac{z(r)}{1+\beta_n} \left( e^{-z(r)} - e^{-n} \right) \\
  &=
  \frac{1 + \beta_n - (z'(r) - z(r))e^{-z(r)} - z(r) e^{-n}}{1+\beta_n} \\
  &=
  \frac{1 + \beta_n - (1 + \beta_n e^{z(r)}) e^{-z(r)} - z(r) e^{-n}}{1 + \beta_n} \\
  &=
  \frac{1 - e^{-z(r)} - z(r) e^{-n}}{1+\beta_n}
  =
  \frac{1 - e^{-q} - q e^{-n}}{1+\beta_n}
\end{align*}
which matches the equation claimed in the lemma statement.
\end{proof}
Compared to the equation~\eqref{exp}
asserted by \Cref{mysterious-useful},
the inequality \eqref{expn} that we need
to prove differs in three important
respects.
\begin{enumerate}
\item The $\exp_n$ function appearing on both sides
  of~\eqref{expn} has been replaced by the
  exponential function in~\eqref{exp}.
\item The left side of~\eqref{exp} has an
  additional $q e^{-n}$ term that does not
  correspond to any term on the left side
  of~\eqref{expn}.
\item The constant on the left side of~\eqref{exp}
  is $\frac{1}{1+\beta_n}$ instead of
   $\alpha_n$. In fact,
  we will be choosing $\alpha_n$
  to be very close to
  $\frac{1}{1+\beta_n}$,
  but slightly smaller.
  The difference between $\alpha_n$
  and $\frac{1}{1+\beta_n}$ will
  be necessary to compensate for
  error terms arising from substituting
  $\exp$ for $\exp_n$ and from the
  additional $q e^{-n}$ term on the left side.
\end{enumerate}
The process of accounting for these differences
between~\eqref{expn} and~\eqref{exp}
begins with the following technical lemma that
provides an upper bound on the value of
$Z(r)$ for $0 \le r \le 1$.
\begin{lem} \label{Z-ub}
  If $n \ge 3$ then $Z(r) < 2$ for all $r \in [0,1]$.
\end{lem}
\begin{proof}
  First, recall from equation~\eqref{eq:integral-z} that
  $
    \int_0^n \frac{dz}{1+z+\beta_n e^z} = 1 .
  $
  As a function of $\beta$, the value of
  $\int_0^n \frac{dz}{1+z+\beta e^z}$ is
  monotonically decreasing, so the fact that
  $\int_0^3 \frac{dz}{1+z+e^z/5} > 1$
  (as can be verified numerically) implies
  that $\beta_3 > 1/5$. Since $\beta_n$ is
  increasing in $n$, this also implies
  $\beta_n > 1/5$ for all $n \ge 3$.

  Now, as $Z(r)$ is increasing in $r$, to
  prove the inequality $Z(r) < 2$ it suffices
  to show that $Z(1) < 2$. Recalling
  equation~\eqref{eq:exp-Z}, we have
  \[
    e^{Z(1)} = \frac{(1+\beta_n) e^{z(1)}}{z'(1)}
      = \frac{(1+\beta_n)e^n}{1+n+\beta_n e^n}
      < \frac{1+\beta_n}{\beta_n} < 6
  \]
  where the final inequality used the fact that $\beta_n > 1/5$.
  Taking the logarithm of both sides, we find that
  $Z(1) < 2$ as desired.
\end{proof}
The next lemma bounds the multiplicative error in using
$\exp_n(-Z(r))$ to approximate $\exp(-Z(r))$.
\begin{lem} \label{expn-approx}
  For $0 \le \lambda < 2$ and $n \ge 2$,
  \begin{equation} \label{eq:expn-approx}
%    \exp_n(-\lambda) \le \exp(-\lambda) < (1-2/n)^{-2} \exp_n(-\lambda) .
    \left( 1 - \tfrac2n \right)^2 \exp(-\lambda) \le \exp_n(-\lambda) \le
    \exp(-\lambda) .
  \end{equation}
\end{lem}
\begin{proof}
  Recalling the definition of the $\exp_n$ function and using
  the inquality $1-x \le e^{-x}$, we have
  $\exp_n(-\lambda) = (1 - \lambda/n)^n \le \exp(-\lambda/n)^n
  = \exp(-\lambda)$ which furnishes the first inequality in the
  lemma statement. To prove the second inequality, note that
  \begin{align*}
    \left( 1 - \tfrac2n \right)^{-2} \exp_n(-\lambda) & >
    \left( 1 - \tfrac{\lambda}{n} \right)^{-2} \exp_n(-\lambda)
     =
    \left( 1 - \tfrac{\lambda}{n} \right)^{n-2}
     =
    \left( 1 + \tfrac{\lambda}{n-\lambda} \right)^{2-n} \\
    & \ge
    \exp \left( - \left( \tfrac{n-2}{n-\lambda} \right) \lambda \right) \\
    & \ge
    \exp \left( - \lambda \right)
  \end{align*}
  where the first and last inequalities made use of the fact
  that $\lambda < 2$, and the second inequality
  % $(1 + \tfrac{\lambda}{n-\lambda})^{2-n} \ge
  % \exp \left( - \left( \tfrac{n-2}{n-\lambda} \right) \lambda \right) $
  depends on the fact that $n \ge 2$, so that
  the exponent $2-n$ is non-positive.
\end{proof}
The following simple inequality will be used twice in the sequel.
\begin{lem} \label{lambda}
  For all $\lambda \ge 0$,
  \begin{equation} \label{eq:lambda}
    \lambda^2 \le e^{\lambda} - 1 .
  \end{equation}
\end{lem}
\begin{proof}
  The function $e^{\lambda} - 2 \lambda$ attains its minimum value
  at $\lambda = \ln(2)$, where its value is strictly positive. Since
  this function is the derivative of $e^{\lambda} - \lambda^2 - 1$, we may
  conclude that the latter function is strictly increasing, and
  therefore $\lambda^2 \le e^{\lambda} - 1$ for all $\lambda \ge 0$.
\end{proof}
Combining the steps above, we can obtain a useful lower bound on the right
side of~\eqref{expn}.
\begin{lem}
  For $n \ge 3$,
  \begin{equation} \label{integral-term}
    \int_0^1 \min\{z(t),q\} \exp_n(-Z(t)) \, dt \ge
    \left(1 - \tfrac5n \right) \frac{1}{1+\beta_n} \left(1 - e^{-q} \right) .
  \end{equation}
\end{lem}
\begin{proof}
  \Cref{Z-ub} justifies applying \Cref{expn-approx} to the factor
  $\exp_n(-Z(t))$ in the integrand, which implies
  \begin{align} \nonumber
    \int_0^1 \min\{z(t),q\} \exp_n(-Z(t)) \, dt & \ge
     \left( 1 - \tfrac2n \right)^2 \int_0^1 \min\{z(t),q\} \exp(-Z(t)) \, dt \\
     & = \left( 1 - \tfrac2n \right)^2 \frac{1}{1+\beta_n}
        \left( 1 - e^{-q} - q e^{-n} \right)
    \label{eq:integral-term.0}
  \end{align}
  where the second line is derived from \Cref{mysterious-useful}.

  Now, the function $\frac{q}{1 - e^{-q}}$ is an increasing function
  of $q$ in the range $0 < q < \infty$, as can be verified by computing
  that the derivative is $e^{-q} \cdot (e^q - 1 - q) \cdot (1 - e^{-q})^{-2}$
  and observing that all three factors are strictly positive. Hence,
  for $0 < q < n$ we have
  \begin{equation} \label{eq:integral-term.1}
    \frac{q}{1-e^{-q}} < \frac{n}{1-e^{-n}} < \frac{e^n}{n}  ,
  \end{equation}
  where the second inequality makes use of \Cref{lambda}.
  Rearranging terms in~\eqref{eq:integral-term.1} we obtain
  \begin{equation} \label{eq:integral-term.2}
    q e^{-n} < \tfrac1n (1 - e^{-q})
  \end{equation}
  and substituting this into the right side
  of~\eqref{eq:integral-term.0} we obtain
  \begin{equation} \label{eq:integral-term.3}
    \int_0^1 \min\{z(t),q\} \exp_n(-Z(t)) \, dt \ge
    \left( 1 - \tfrac2n \right)^2 \frac{1}{1+\beta_n}
    \left( 1 - e^{-q} \right) \left( 1 - \tfrac1n \right) .
  \end{equation}
  The observation that $(1 - \frac2n)^{2} (1 - \frac1n) > (1 - \frac5n)$
  concludes the proof.
\end{proof}
Next we need a lemma accounting for the discrepancy between
the function $1 - \exp_n(-q)$ appearing on the left side
of~\eqref{expn} and the function $1-e^{-q}$
in~\eqref{integral-term}.
\begin{lem} \label{almost-there}
  For $0 < q < n$,
  \begin{equation} \label{at.0}
    \left(1 - \tfrac1n \right) \cdot \left( 1 - \exp_n(-q) \right) < 1 - e^{-q} .
  \end{equation}
\end{lem}
\begin{proof}
  If $q \ge n/2$ then $\exp_n(-q) < e^{-q} \le e^{-n/2} < \frac1n$,
  so both $1 - \exp_n(-q)$ and $1 - e^{-q}$ are in the interval
  $\left[ 1 - \frac1n, 1 \right]$, and the lemma follows easily
  in this case. Assume henceforth that $q < n/2$.

  For $0 < s < n$ the inequality $\ln(1-\frac{s}{n}) < -\frac{s}{n}$
  implies that $-1-\ln(1-\frac{s}{n}) > -1 + \frac{s}{n}$.
  Integrating from $s=0$ to $s=q$, we derive
  \begin{equation} \label{at.1}
    (n-q) \ln \left( 1 - \frac{q}{n} \right) >
    -q + \frac{q^2}{2n} = - \left( \frac{n-q}{n} \right) q - \frac{q^2}{2n} .
  \end{equation}
  Multiplying both sides by $\frac{n}{n-q}$ we obtain
  \begin{equation} \label{at.2}
    n \ln \left( 1 - \frac{q}{n} \right) >
    -q - \frac{q^2}{2(n-q)} .
  \end{equation}
  Now exponentiate both sides of~\eqref{at.2}.
  \begin{align}
    \nonumber
    \exp_n(-q) & > e^{-q} \exp \left( -\frac{q^2}{2(n-q)} \right) >
      e^{-q} \left[ 1 - \frac{q^2}{2(n-q)} \right] \\
    1 - \exp_n(-q) &< 1 - e^{-q} + \frac{q^2 e^{-q}}{2(n-q)} <
      1 - e^{-q} + \frac{q^2 e^{-q}}{n}
    \label{at.3}
  \end{align}
  where the final inequality used the fact that $q < n/2$.
  \Cref{lambda}
  implies that $q^2 e^{-q} < 1 - e^{-q}$ so
  \[
    1 - \exp_n(-q) < 1 - e^{-q} + \frac{1 - e^{-q}}{n} =
      (1 + \tfrac1n) \cdot (1 - e^{-q}) <
      (1 - \tfrac1n)^{-1} \cdot (1-e^{-q})
  \]
  and the lemma follows upon rearranging terms.
\end{proof}
Combining equation~\eqref{integral-term} with \Cref{almost-there},
we obtain the inequality
\[
  \left( 1 - \tfrac1n \right)
  \left( 1 - \tfrac5n \right)
  \left( 1 + \beta_n \right)^{-1}
  (1 - \exp_n(-q)) <
  \int_0^1 \min\{z(t),q\} \exp_n(-Z(t)) \, dt
\]
which proves~\eqref{expn} with
\begin{equation} \label{alphan}
  \alpha_n = \left(1 - \tfrac6n \right) \left(1 + \beta_n \right)^{-1}
\end{equation}
since $(1 - \frac1n)(1 - \frac5n) > 1 - \frac6n$.

Finally, let us bound the difference between $\alpha_n$
and the constant $\alpha = 0.745\dots$ defined by the
equation
$$\int_0^1 \frac{dy}{ y - y \ln y + 1 - \tfrac{1}{\alpha}}
= 1.$$
Perform the substitution $y = e^{-z}$
to deduce that
$$\int_0^\infty \frac{dz}{1 + z + \left( \tfrac{1}{\alpha} - 1 \right) e^z}.$$
Set $\beta = \tfrac{1}{\alpha} - 1$, which implies
$\alpha = (1+\beta)^{-1}$.
Noting that
$$
  \int_0^n \frac{dz}{1+z+\beta e^z} <
  \int_0^\infty \frac{dz}{1+z+\beta e^z} = 1
  = \int_0^n \frac{dz}{1+z+\beta_n e^z}
$$
and that $\int_0^n \frac{dz}{1+z+x e^z}$ is a
monotonically decreasing function of $x$,
we deduce that $\beta_n < \beta$ which implies
$(1+\beta_n)^{-1} > (1+\beta)^{-1} = \alpha$.
Hence $\alpha_n > \left(1 - \tfrac6n \right) \alpha$.

\section{Tightness of bounds}
\label{appdx-tight}

In this section we prove that the bounds in
each part of \Cref{distributional bounds} are
tight under their respective assumptions.

\begin{prop}
\label{1/2-neg-corr}
  For any $\eps>0$ there exists a joint distribution of
  the agent's
  and principal's utilities such that
  the principal's expected utility is always less
  than $\left( \frac12 + \eps \right) \expect [x_*]$
  regardless of the choice of mechanism.
\end{prop}
\begin{proof}
  For a parameter $H \gg 1/\eps$,
  define the joint distribution of $(x,y)$
  to be a mixture of uniform distributions on two
  rectangles. The first is the rectangle
  $Q_1=[1-\frac1H,1+\frac1H] \times [2,3]$,
  the second is the rectangle $Q_2 = [H,H+2] \times [0,1]$.
  The pair $(x,y)$
  is drawn from the uniform distribution
  on $Q_1$ with probability $1-\frac{1}{nH}$ and from the uniform
  distribution on $Q_2$ with probability $\frac{1}{nH}$.
  When we draw $n$ i.i.d.\ samples $\{(x_i,y_i)\}_{i=1}^n$
  from this joint distribution, the random variable $\max \{ x_i \}$
  is uniformly distributed in $[1-\frac1H,1+\frac1H]$ with
  probability $(1-\frac{1}{nH})^n = 1 - \frac1H + O(H^{-2})$, and
  it is uniformly distributed in $[H,H+2]$ with probability
  $1-(1-\frac{1}{nH})^n = \frac1H - O(H^{-2})$. Therefore,
  letting $x_* = \max \{x_i \mid 1 \le i \le n\}$,
  \[
    \expect x_* =
    1 \cdot (1 - \tfrac{1}{nH})^n \;\; + \;\;
    (n+1) \cdot [1 - (1-\tfrac{1}{nH})^n] =
    1 + n \cdot \left( \tfrac1H - O(H^{-2}) \right) =
    2 - O(\tfrac1H) .
  \]
  Now we bound the principal's expected utility under any
  mechanism. By \Cref{spm} it is enough to consider single
  proposal mechanisms. Let $R$ denote the eligible set of
  the mechanism, and let $p$ denote the probability that
  a random sample from the distribution of $(x,y)$ belongs
  to $Q_1 \cap R$. Since
  the agent prefers every element of $Q_1$ to every element
  of $Q_2$, if even one of the $n$ samples belongs to $Q_1 \cap R$
  then the agent will propose an element of $Q_1 \cap R$. This
  event $\mathcal{E}_1$ happens with probability $1 - (1-p)^n$ and
  the principal's expected utility conditional on $\mathcal{E}_1$
  is bounded above by $1 + \frac1H$. The other case in which the
  principal gets non-zero utility is if all $n$ samples belong
  to the complement of $Q_1 \cap R$, but they don't all belong
  to the complement of $Q_2 \cup (Q_1 \cap R)$. The probability
  of this event $\mathcal{E}_2$ is
  $(1-p)^n - (1-p-n^{-2})^n = \frac{1}{H} (1-p)^{n-1} + O(H^{-2})$,
  and principal's expected utility conditional on $\mathcal{E}_2$
  is bounded above by $H+2$. Hence, denoting the principal's utility
  by $\hat{x}$,
  \begin{align*}
    \expect \hat{x} &\le
    \Pr(\mathcal{E}_1) \cdot \left( 1 + \tfrac1H \right)
    +
    \Pr(\mathcal{E}_2) \cdot (H+2) \\
    &=
    [1 - (1-p)^n] \left( 1 + \tfrac1H \right) +
    \left[ \tfrac{1}{H} (1-p)^{n-1} + O(H^{-2}) \right]
    (H+2) \\
    &=
    1 - (1-p)^n + (1-p)^{n-1} + O \left( \tfrac1H \right)
  \end{align*}
  The maximum of $(1-p)^n + (1-p)^{n-1}$ is
  $1 + \frac1n (1-\frac1n)^{n-1})$, achieved
  at $p=\frac1n$. Hence, $\expect \hat{x} = 1 + O(\frac1n) + O(\frac1H)$
  which is less than $(\frac12 + \eps) \expect x_*$
  provided $n$ and $H$ are both sufficiently large.
\end{proof}

The remaining results in this section pertain to the
case in which the principal's and agent's utilities
are independent random variables.

% \begin{lem}
%   It is without loss of generality to assume $y$ is
%   independent of both $x$ and $M(x,y)$.
% \end{lem}
%
% \begin{prop}
% \label{1/2-atom}
%   For any $\eps>0$ there exists a product distribution of
%   the agent's and principal's utilities such that if the
%   principal runs any mechanism
%   that can be implemented when the
%   agent's utility is unobservable,
%   then she cannot guarantee a utility
%   greater than
%   $\left( \frac12 + \eps \right) \expect [x_*]$.
% \end{prop}

\begin{prop}
\label{1-1/e-lb}
  For any $\eps>0$ there exists a product distribution
  of the agent's and principal's utilities such that
  the marginal distribution of each party's utility is
  atomless, and if the principal cannot observe the agent's
  utility then
  the principal's expected utility is always less
  than $\left( 1 - \frac1e + \eps \right) \expect [x_*]$
  regardless of the choice of mechanism.
\end{prop}
\begin{proof}
  Let $\Omega = [0,1] \times [0,1]$ under the uniform
  measure. For a sample point $(w,z) \in \Omega$ we
  will be defining the principal's utility to be $x=f(w)$,
  where the function $f$ is defined as follows.
  For some $H \gg 1/\eps$ let $f$ map the
  interval $[0,1 - \frac{1}{(e-2)nH}]$ to
  $[1, 1 + \frac{1}{H}]$ and let $f$ map
  interval $[1-\frac{1}{(e-2)nH},1]$ to
  $[H+1,H+1+\frac{1}{H}]$. The restriction
  of $f$ to each of the two subintervals is
  linear. Thus, the distribution of $x$ is a
  mixture of two uniform distributions, on
  the intervals $[1,1+\frac{1}{H}]$ and
  $[H+1,H+1+\frac{1}{H}]$, with the mixture
  weights being $1 - \frac{1}{(e-2)nH}$ and
  $\frac{1}{(e-2)nH}$, respectively.

  By \Cref{spm} we can restrict our attention
  to single proposal mechanisms.
  Let $R$ denote the eligible set of the mechanism,
  and let $\chi_R$ denote its characteristic function.

  The agent's utility $y$ will be defined as
  a function of both $w$ and $z$, but with the
  property that conditional on the values of
  $w$ and $\chi_R(w,z)$, the value of $y$ is always
  uniformly distributed. Such a function $y$ can
  be constructed as follows. For $w \in [0,1]$
  let $R_w$ denote the intersection of $R$ with
  $\{w\} \times [0,1]$; this set is measurable
  for almost every $w$. If the measure of $R_w$,
  $m(R_w)$,
  belongs to $\{0,1\},$ then $y(w,z)=z$. If
  $0 < m(R_w) < 1$ then
  \[
    y(w,z) = \begin{cases}
      \frac{m(R_w \cap [0,z])}{m(R_w)} &
      \mbox{if } z \in R_w \\[1ex]
      \frac{m(\overline{R_w} \cap [0,z])}{m(\overline{R_w})} &
      \mbox{if } z \in \overline{R_w} .
      \end{cases}
  \]

  Let $p = \frac{1}{(e-2)nH}$ denote the probability
  that a random $(w,z)$ satisfies $x(w) \ge H+1$.
  Let $q = m(R)$ be the probability
  that a random sample $(w,z) \in [0,1]^2$ belongs
  to the eligible set. Conditional on belonging to the
  eligible set, a random sample point $(w,z)$ has
  probability at most $p/q$ of satisfying $x(w) \ge H+1$.
  Furthermore, since $y(w,z)$ is independent of both
  $w$ and $\chi_R(w,z)$, conditional on the set of eligible
  points sampled by the agent being non-empty, the
  conditional distribution of $w$ and $\chi_R(w,z)$
  for the agent's proposal (assuming best response
  to the mechanism) is equal to the conditional
  distribution of a single random sample $(w,z) \in [0,1]^2$
  conditioned on belonging to $R$. Hence, if $\hat{x}$ denotes
  the value of $x(w,z) \cdot \chi_R(w,z)$ when the agent's proposal $(w,z)$ is
  a best response to the mechanism, we have
  \[
     \expect \hat{x} = (1-(1-q)^n) \cdot \left[
         1 + \tfrac{p}{q} \cdot H + O(\tfrac1H)
       \right] .
  \]
  Meanwhile, if $x_* = \max_i \{x_i\}$ then
  \[
      \expect x_* = 1 + (1 - (1-p)^n) H + O(\tfrac1H) .
  \]
  Let $q = \phi/n$.
  Ignoring $O(\tfrac1H)$ terms that vanish as $H \to \infty$
  and $O(1/n)$ terms that vanish as $n \to \infty$ we have
  \begin{align}
  \nonumber
    \expect x_* & \approx 1 + \tfrac{1}{(e-2)H} \cdot H = \tfrac{e-1}{e-2} \\
  \nonumber
    \expect \hat{x} & \approx \left(1 - e^{-\phi} \right) \cdot \left(1 + \tfrac{1}{(e-2)\phi}
    \right) \\
  \label{eq:1-1/e.0}
    \frac{\expect \hat{x}}{\expect x_*} & \approx \left(1 - e^{-\phi} \right) \cdot \left( \frac{e-2}{e-1} + \frac{1}{(e-1)\phi} \right).
  \end{align}
  The right side of~\eqref{eq:1-1/e.0} is maximized at $\phi=1$, where it equals $1 - \frac1e$.
  Thus, for any mechanism that is obvlivious to the value of $y$,
  $\expect[\hat{x}] / \expect[x_*] \le 1 - \frac1e + O(1/H) + O(1/n)$, which is less
  than $1 - \frac1e + \eps$ for $H$ and $n$ sufficiently large.
\end{proof}

\begin{prop}
\label{0.745-lb}
  Let $\alpha=0.745\dots$ denote the solution to
  $\int_0^1 \frac{dy}{y + y \ln y -  1/\alpha + 1} = 1$.
  For any $\eps>0$ there exists a product distribution
  of the agent's and principal's utilities such that
  the principal's expected utility is always less than
  $(\alpha+\eps) \, \expect [ x_* ]$, regardless of the
  choice of mechanism.
\end{prop}
\begin{proof}
  The proof is by contradiction. It is known
  \citep{hill-kertz,kertz86} that there is no
  prophet inequality with factor $\alpha+\eps$
  for i.i.d.\ distributions. In other words,
  there exists a distribution $F$ such that
  if $x_1,\ldots,x_n$ and i.i.d.\ with distribution
  $F$, then for every stopping rule $\tau$ we have
  \begin{equation} \label{stop-rule}
    \expect x_\tau < (\alpha + \eps) \, \expect x_* .
  \end{equation}
  To prove the proposition we will consider the
  sample space $\Omega = \reals_+ \times [0,1]$,
  with sample points denoted by $(x,y)$,
  under the product distribution where $x$ is distributed
  according to $F$ and $y$ is uniformly dstributed
  in $[0,1]$. We will show how to
  transform any mechanism that guarantees utility
  at least $(\alpha+\eps) \, \expect[ x_* ]$ for the
  principal into a stopping rule $\tau$ that
  violates~\eqref{stop-rule}, yielding the desired
  contradiction.

  The transformation is simple to describe. If there
  is a mechanism guaranteeing utility
  at least $(\alpha+\eps) \, \expect[ x_* ]$ for the
  principal, we can assume it is a single proposal
  mechanism by \Cref{spm}. Let $R$ denote the
  eligible set of the mechanism. The stopping rule
  $\tau$ is specified as follows. First, it draws
  $n$ i.i.d.\ uniform random samples from $[0,1]$
  and sorts them into decreasing order, numbering
  them as $y_1 \ge y_2 \ge \cdots \ge y_n$.
  Then the stopping time $\tau$ is defined as the
  minimum $i$ such that $(x_i,y_i) \in R$.

  % The analysis of this stopping rule depends on
  % two key observations:
  % \begin{enumerate}
  %   \item
    Observe that
    the distribution of the sequence of pairs
    $(x_1,y_1),\ldots,(x_n,y_n)$ obtained when
    running the stopping rule is identical to
    the distribution over sequences that one
    obtains by drawing $n$ i.i.d.\ samples from
    $\Omega$ and arranging them in order of
    decreasing $y$. This is because the
    random variables $x_1,x_2,\ldots,x_n$
    are i.i.d., so the distribution of the
    random sequence $x_1,\ldots,x_n$ is
    invariant under permutations. Thus, the
    process of first pairing
    $x_i$ with $y_i$ and then sorting by $y_i$
    generates the same distribution over
    sequences of pairs as if one first
    sorts $y_1,\ldots,y_n$ and then pairs
    $x_i$ with $y_i$.

    Next, observe that when the sequence
    of pairs $(x_1,y_1),\ldots,(x_n,y_n)$
    is drawn as in the preceding paragraph,
    both the stopping rule and the agent's
    best response to the single proposal mechanism
    select the pair
    $(x_i,y_i) \in R$ with maximum $y_i$, or
    they select $\bot$ if there is no pair
    $(x_i,y_i) \in R$. Hence, the principal's
    expected payoff when running the mechanism is equal to
    $\expect x_\tau$. Our assumption that
    the mechanism guarantees at least
    $(\alpha+\eps) \, \expect x_*$ to the
    principal thus implies that
    $\expect x_\tau \ge (\alpha+\eps) \, \expect x_*$,
    contradicting~\eqref{stop-rule}.
\end{proof}

\end{document}